\documentclass[english]{article}
\usepackage[T1]{fontenc}
\usepackage[latin9]{inputenc}
\usepackage{geometry}
\usepackage{color}
\usepackage{babel}
\usepackage{array}
\usepackage{float}
\usepackage{mathtools}
\usepackage{multirow}
\usepackage{amsmath}
\usepackage{amsthm}
\usepackage{amssymb}
\usepackage{setspace}
\usepackage[unicode=true,pdfusetitle,
 bookmarks=true,bookmarksnumbered=false,bookmarksopen=false,
 breaklinks=false,pdfborder={0 0 0},pdfborderstyle={},backref=false,colorlinks=true]
 {hyperref}

\makeatletter

\providecommand{\tabularnewline}{\\}

\theoremstyle{definition}
\newtheorem{defn}{\protect\definitionname}
\theoremstyle{definition}
 \newtheorem{example}{\protect\examplename}
\theoremstyle{plain}
\newtheorem{thm}{\protect\theoremname}
\theoremstyle{plain}
\newtheorem{lem}{\protect\lemmaname}
\theoremstyle{plain}
\newtheorem{prop}{\protect\propositionname}
\theoremstyle{plain}
\newtheorem{assumption}{\protect\assumptionname}
\theoremstyle{remark}
\newtheorem{rem}{\protect\remarkname}
\theoremstyle{plain}
\newtheorem{cor}{\protect\corollaryname}
\theoremstyle{remark}
\newtheorem*{claim*}{\protect\claimname}

\hypersetup{urlcolor=blue}
\usepackage[bottom,multiple]{footmisc}

\@ifundefined{showcaptionsetup}{}{%
 \PassOptionsToPackage{caption=false}{subfig}}
\usepackage{subfig}
\makeatother

\providecommand{\assumptionname}{Assumption}
\providecommand{\claimname}{Claim}
\providecommand{\corollaryname}{Corollary}
\providecommand{\definitionname}{Definition}
\providecommand{\examplename}{Example}
\providecommand{\lemmaname}{Lemma}
\providecommand{\propositionname}{Proposition}
\providecommand{\remarkname}{Remark}
\providecommand{\theoremname}{Theorem}

\begin{document}
\textwidth=6.1in
\textheight=8.6in
\setcounter{secnumdepth}{2}
\pagenumbering{gobble}
\thispagestyle{empty} 
\title{Matching Multidimensional Types: Theory and Application}
\author{Veli Safak\thanks{Carnegie Mellon University Qatar, e-mail: vsafak@andrew.cmu.edu}\thanks{I would like to express my gratitude to my advisor Axel Z. Anderson
for his guidance, strong and continuous support, and always having
faith in my capability. I would also like to thank James W. Albrecht
and Luca Anderlini for their constructive comments. I would further
like to thank Laurent Bouton, Dan Cao, Chris Chambers, Roger Lagunoff,
Arik M. Levinson, Yusufcan Masatlioglu, Franco Peracchi, and John
Rust for helpful discussions. }}
\date{6/10/2020}
\maketitle
\begin{abstract}
\noindent Becker (1973) presents a bilateral matching model in which
scalar types describe agents. For this framework, he establishes the
conditions under which positive sorting between agents' attributes
is the unique market outcome. Becker's celebrated sorting result has
been applied to address many economic questions. However, recent empirical
studies in the fields of health, household, and labor economics suggest
that agents have multiple outcome-relevant attributes. In this paper,
I study a matching model with multidimensional types. I offer multidimensional
generalizations of concordance and supermodularity to construct three
multidimensional sorting patterns and two classes of multidimensional
complementarities. For each of these sorting patterns, I identify
the sufficient conditions which guarantee its optimality. In practice,
we observe sorting patterns between observed attributes that are aggregated
over unobserved characteristics. To reconcile theory with practice,
I establish the link between production complementarities and the
aggregated sorting patterns. Finally, I examine the relationship between
agents' health status and their spouses' education levels among U.S.
households within the framework for multidimensional matching markets.
Preliminary analysis reveals a weak positive association between agents'
health status and their spouses' education levels. This weak positive
association is estimated to be a product of three factors: (a) an
attraction between better-educated individuals, (b) an attraction
between healthier individuals, and (c) a weak positive association
between agents' health status and their education levels. The attraction
channel suggests that the insurance risk associated with a two-person
family plan is higher than the aggregate risk associated with two
individual policies.
\end{abstract}
\pagenumbering{arabic} 
\setcounter{page}{1}

\section*{Introduction\label{sec:Introduction}}

Becker (1973) proposes a general framework for two-sided frictionless
matching models in which scalar types represent the agents on each
side of the market, i.e. each agent has only one outcome-relevant
attribute. A match between two agents (one from each side) generates
a type-dependent matching output. A social planner\footnote{A decentralized version of this model with perfectly transferable
utilities can easily be constructed by using the dual version of the
planner's problem.} maximizes the aggregate output by matching the agents in pairs. There
are two essential components of Becker's theory: complementarity and
sorting. The matching output exhibits \emph{strictly positive complementarity}
when the marginal product of an agent strictly increases in his/her
partner's type. Similarly, the matching output exhibits \emph{strictly
negative complementarity} when the marginal product of an agent strictly
decreases in his/her partner's type. 

Becker (1973) shows that if the matching output exhibits strictly
positive complementarity, then the unique solution to the planner's
problem is \emph{positive sorting}, i.e. the highest types are matched
together, then the next highest types, etc. Likewise, negative sorting
is the unique solution when the matching output exhibits strictly
negative complementarity. Economists have applied Becker's assortative
matching results\footnote{For a recent literature review on the matching markets; I refer the
readers to Chade et al. (2017).} to address several questions. For example, Kremer (1993) sheds light
on the positive correlation between wages of the workers within a
firm. Gabaix and Landier (2008) explain the rise in the CEOs' salaries
and its connection to the increase in firms' sizes over time. 

In many applications, the agents may have multiple outcome-relevant
attributes. For instance, education and race in the dating/marriage
market, workers' social and cognitive skills in the labor market,
and doctors' listening skills for diagnosis and fostering the doctor-patient
relationship in the healthcare market are some well-documented outcome-relevant
attributes in literature. In the next section, I survey additional
recent studies that support the presence of multiple outcome-relevant
attributes. If a single index can capture all outcome-relevant information,
then a unidimensional model may be suitable. However, the single index
assumption is implausible in many applications. For example, Chiappori
et al. (2012) analyze the U.S. marriage market by using a multidimensional
matching model with an index restriction, i.e. two agents with different
attributes are identical if they have the same index value calculated
by an exogenous index function. Fletcher and Padron (2015) provide
empirical evidence against the implications of Chiappori et al.'s
(2012) single index assumption. 

The empirical support in health, household and labor economics for
multidimensional types highlight the practical importance of the multidimensional
matching theory. In this paper, I present a matching model with multidimensional
types and examine the link between output complementarities and sorting
patterns. The only difference between Becker's framework and the framework
presented in this paper is that I allow the agents to have multiple
outcome-relevant attributes. Although the proposed model is general,
to ease the exposition throughout the introduction, I consider a particular
labor market model in which firms and workers have only two scalar
productive skills: cognitive and social. In this context, a matching
distribution satisfies \emph{global positive sorting }if and only
if it exhibits positive sorting (a) between firms' and their workers'
cognitive skills, and (b) between firms' and their workers' social
skills. Similarly, a matching distribution satisfies \emph{global
negative sorting }if and only if it exhibits negative sorting (a)
between firms' and their workers' cognitive skills, and (b) between
firms' and their workers' social skills. 

A naive application of Becker's sorting result implies positive sorting
between firms' and their workers' cognitive skills when the marginal
product of each firm's cognitive skill strictly increases in its worker's
cognitive skill. Similarly, positive sorting between firms' and their
workers' social skills is obtained when the marginal product of each
firm's social skill strictly increases in its worker's social skill
according to Becker's sorting result. In a multidimensional matching
market, one may observe simultaneous positive complementarities between
cognitive skills and social skills. However, simultaneous positive
sorting between cognitive skills and social skills may not be feasible. 

Consider two firms, $x=\left(x_{c},x_{s}\right)$ and $x^{\prime}=\left(x_{c}^{\prime},x_{s}^{\prime}\right)$,
such that $x=\left(10,10\right)$ and $x^{\prime}=\left(20,20\right)$.
Furthermore, suppose that there are two workers, $y=\left(y_{c},y_{s}\right)$
and $y^{\prime}=\left(y_{c}^{\prime},y_{s}^{\prime}\right)$, such
that $y=\left(10,20\right)$ and $y^{\prime}=\left(20,10\right)$.
Notice that matching $x$ with $y$ and $x^{\prime}$ with $y^{\prime}$
satisfies positive sorting between cognitive skills and violates positive
sorting between social skills. Similarly, matching $x$ with $y^{\prime}$
and $x^{\prime}$ with $y$ satisfies positive sorting between social
skills and violates positive sorting between cognitive skills. In
this paper, I show that, conditioning on the existence, Becker's sorting
results apply: when the output function exhibits strictly positive
complementarities between cognitive skills and between manual skills,
if there exists a matching scheme which satisfies global positive
sorting, then (a) every optimal matching scheme satisfies global positive
sorting, and (b) every matching scheme satisfying global positive
sorting solves the planner's problem. I establish the optimality of
global sorting for a general global sorting class in Proposition \ref{prop:1}.

Since global positive sorting may not be feasible, I examine an alternative
sorting pattern inspired by Chiappori et al. (2017). Chiappori et
al. (2017) study a marriage model in which one continuous variable
(socioeconomic status) and one binary variable (smoking habit) represent
the agents on each side of the market. They categorize couples into
two main groups. In the first group, both men and women are non-smokers.
In the second group, at least one of the spouses smokes. They assume
that the matching output of a couple is the multiplication of spouses'
socioeconomic status. If there is a smoker in the household, then
the output is scaled down by a constant. Under this complementarity
structure, Chiappori et al. (2017) predict positive sorting between
agents' and their spouses' socioeconomic status within each group.
Notice that one can easily apply the idea of splitting the sample
into different groups and studying the sorting patterns for each group
in a more general setting.

For the previous labor market example, a matching satisfies \emph{within-group
positive sorting between cognitive skills} if it exhibits positive
sorting between firms' and their workers' cognitive skills for all
social skill pairs of firms and workers $\left(x_{s},y_{s}\right)$.
Consider four firms and four workers: $\left\{ \left(10,10\right),\left(10,20\right),\left(20,10\right),\left(20,20\right)\right\} $.
Here, matching the $\left(10,\boldsymbol{10}\right)$ firm with the
$\left(10,\boldsymbol{10}\right)$ worker and the $\left(20,\boldsymbol{10}\right)$
firm with the $\left(20,\boldsymbol{10}\right)$ worker is consistent
with within-group positive sorting between cognitive skills for the
$\left(\boldsymbol{10},\boldsymbol{10}\right)$ social skill combination.
Within-group sorting solves the feasibility problem: for arbitrary
distributions of agents, there exists a matching scheme which satisfies
(a) within-group positive sorting between cognitive skills and (b)
within-group positive sorting between social skills. Furthermore,
I show that, when the matching output exhibits strictly positive complementarities
$\left(\clubsuit\right)$ between cognitive skills and $\left(\spadesuit\right)$
between social skills, every optimal matching distribution satisfies
(a) within-group positive sorting between cognitive skills and (b)
within-group positive sorting between social skills. I establish the
optimality of within-group sorting for a general within-group sorting
class in Proposition \ref{prop:1}.

Within-group sorting has two major drawbacks as a sorting concept.
First of all, there may be multiple ways to match agents without violating
within-group sorting. 
\noindent \begin{center}
{\scriptsize{}}%
\begin{tabular}{c|cccccccc|ccccc}
\multirow{2}{*}{{\scriptsize{}Firms}} & \multirow{2}{*}{{\scriptsize{}$\left(10,10\right)$}} & \multirow{2}{*}{{\scriptsize{}$\left(10,20\right)$}} & \multirow{2}{*}{{\scriptsize{}$\left(20,10\right)$}} & \multirow{2}{*}{{\scriptsize{}$\left(20,20\right)$}} &  &  &  & \multirow{2}{*}{{\scriptsize{}Firms}} & \multirow{2}{*}{{\scriptsize{}$\left(10,10\right)$}} & \multirow{2}{*}{{\scriptsize{}$\left(10,20\right)$}} & \multirow{2}{*}{{\scriptsize{}$\left(20,10\right)$}} & \multirow{2}{*}{{\scriptsize{}$\left(20,20\right)$}} & \tabularnewline
 &  &  &  &  & \multirow{2}{*}{} &  &  &  &  &  &  &  & \tabularnewline
\cline{1-5} \cline{2-5} \cline{3-5} \cline{4-5} \cline{5-5} \cline{9-13} \cline{10-13} \cline{11-13} \cline{12-13} \cline{13-13} 
\multirow{2}{*}{{\scriptsize{}Matched Worker}} & \multirow{2}{*}{{\scriptsize{}$\left(20,20\right)$}} & \multirow{2}{*}{{\scriptsize{}$\left(10,20\right)$}} & \multirow{2}{*}{{\scriptsize{}$\left(20,10\right)$}} & \multirow{2}{*}{{\scriptsize{}$\left(10,10\right)$}} &  &  &  & \multirow{2}{*}{{\scriptsize{}Matched Worker}} & \multirow{2}{*}{{\scriptsize{}$\left(10,10\right)$}} & \multirow{2}{*}{{\scriptsize{}$\left(10,20\right)$}} & \multirow{2}{*}{{\scriptsize{}$\left(20,10\right)$}} & \multirow{2}{*}{{\scriptsize{}$\left(20,20\right)$}} & \tabularnewline
 &  &  &  &  &  &  &  &  &  &  &  &  & \tabularnewline
\cline{1-5} \cline{2-5} \cline{3-5} \cline{4-5} \cline{5-5} \cline{9-13} \cline{10-13} \cline{11-13} \cline{12-13} \cline{13-13} 
\multicolumn{5}{c}{Matching scheme-1} &  &  &  & \multicolumn{5}{c}{Matching scheme-2} & \tabularnewline
\end{tabular}{\scriptsize\par}
\par\end{center}

\noindent \ 

\noindent Note that for each social skill combination of firms and
workers, there is only one firm-worker couple under these two matching
schemes. The same is also true for each cognitive skill combination
of firms and workers. Consequently, these matching schemes satisfy
within-group positive sorting between cognitive skills and social
skills. 

Secondly, a matching scheme may satisfy within-group positive sorting
and cannot be optimal for any matching output that exhibits strictly
positive complementarities between cognitive skills and between social
skills. For matching scheme-1, a swap between the first and the last
firm-worker couples, i.e. $\left(\left(10,10\right),\left(20,20\right)\right)$
and $\left(\left(20,20\right),\left(10,10\right)\right)$, strictly
increases the aggregate output for any matching output which exhibits
strictly positive complementarities between cognitive skills and between
social skills. Therefore, matching scheme-1 can never be an optimal
matching scheme when the matching output exhibits strictly positive
complementarities between cognitive skills and between social skills. 

To obtain a finer characterization of optimal matching schemes, I
consider another extension of Becker's sorting concepts and propose
a weak sorting notion: a matching scheme satisfies \emph{weak positive
sorting} if there does not exist a pair of matched couples that (a)
is consistent with global negative sorting, and (b) violates global
positive sorting. Notice that the first and the last firm-worker couples
in matching scheme-1 violate global positive sorting. Indeed, the
set of matching schemes that satisfy weak positive sorting is a subset
of the set of matching schemes that satisfy within-group positive
sorting. More importantly, I show that the set of optimal matching
schemes for any matching output that exhibits strictly positive complementarities
between cognitive skills and between social skills is a subset of
the set of matching schemes that satisfies weak positive sorting.
I present a general version of this result in Proposition \ref{prop:1}.

Similar to within-group sorting, a matching scheme may satisfy weak
positive sorting and can never be optimal when the matching output
exhibits strictly positive complementarities between cognitive skills
and between social skills. Note that matching scheme-3 below satisfies
weak positive sorting. At the same time, the following swap sequence
strictly increases the aggregate output for any matching output that
exhibits strictly positive complementarities between cognitive skills
and between social skills:

\noindent Swap-1: Between the first and the second couples

\noindent Swap-2: Between the third and the fourth couples

\noindent Swap-3: Between the first and the last couples matched after
swap-1 and swap-2 
\noindent \begin{center}
\begin{table}[H]
\begin{minipage}[t]{0.4\columnwidth}%
\noindent \begin{center}
{\scriptsize{}}%
\begin{tabular}{c|cccccccc|ccccc}
\multirow{2}{*}{{\scriptsize{}Firms}} & \multirow{2}{*}{{\scriptsize{}$\left(10,10\right)$}} & \multirow{2}{*}{{\scriptsize{}$\left(10,20\right)$}} & \multirow{2}{*}{{\scriptsize{}$\left(20,10\right)$}} & \multirow{2}{*}{{\scriptsize{}$\left(20,20\right)$}} &  &  &  & \multirow{2}{*}{{\scriptsize{}Firms}} & \multirow{2}{*}{{\scriptsize{}$\left(10,10\right)$}} & \multirow{2}{*}{{\scriptsize{}$\left(10,20\right)$}} & \multirow{2}{*}{{\scriptsize{}$\left(20,10\right)$}} & \multirow{2}{*}{{\scriptsize{}$\left(20,20\right)$}} & \tabularnewline
 &  &  &  &  & \multirow{2}{*}{} &  &  &  &  &  &  &  & \tabularnewline
\cline{1-5} \cline{2-5} \cline{3-5} \cline{4-5} \cline{5-5} \cline{9-13} \cline{10-13} \cline{11-13} \cline{12-13} \cline{13-13} 
\multirow{2}{*}{{\scriptsize{}Matched Worker}} & \multirow{2}{*}{{\scriptsize{}$\left(10,20\right)$}} & \multirow{2}{*}{{\scriptsize{}$\left(20,20\right)$}} & \multirow{2}{*}{{\scriptsize{}$\left(10,10\right)$}} & \multirow{2}{*}{{\scriptsize{}$\left(20,10\right)$}} &  &  &  & \multirow{2}{*}{{\scriptsize{}Matched Worker}} & \multirow{2}{*}{{\scriptsize{}$\left(\boldsymbol{20,20}\right)$}} & \multirow{2}{*}{{\scriptsize{}$\left(\boldsymbol{10,20}\right)$}} & \multirow{2}{*}{{\scriptsize{}$\left(10,10\right)$}} & \multirow{2}{*}{{\scriptsize{}$\left(20,10\right)$}} & \tabularnewline
 &  &  &  &  &  &  &  &  &  &  &  &  & \tabularnewline
\cline{1-5} \cline{2-5} \cline{3-5} \cline{4-5} \cline{5-5} \cline{9-13} \cline{10-13} \cline{11-13} \cline{12-13} \cline{13-13} 
\multicolumn{6}{c}{} &  &  & \multicolumn{6}{c}{}\tabularnewline
\multicolumn{5}{c}{Matching scheme-3} &  &  &  & \multicolumn{5}{c}{Swap-1} & \tabularnewline
\multicolumn{1}{c}{} &  &  &  &  &  &  &  & \multicolumn{1}{c}{} &  &  &  &  & \tabularnewline
\multirow{2}{*}{{\scriptsize{}Firms}} & \multirow{2}{*}{{\scriptsize{}$\left(10,10\right)$}} & \multirow{2}{*}{{\scriptsize{}$\left(10,20\right)$}} & \multirow{2}{*}{{\scriptsize{}$\left(20,10\right)$}} & \multirow{2}{*}{{\scriptsize{}$\left(20,20\right)$}} &  &  &  & \multirow{2}{*}{{\scriptsize{}Firms}} & \multirow{2}{*}{{\scriptsize{}$\left(10,10\right)$}} & \multirow{2}{*}{{\scriptsize{}$\left(10,20\right)$}} & \multirow{2}{*}{{\scriptsize{}$\left(20,10\right)$}} & \multirow{2}{*}{{\scriptsize{}$\left(20,20\right)$}} & \tabularnewline
 &  &  &  &  &  &  &  &  &  &  &  &  & \tabularnewline
\cline{1-5} \cline{2-5} \cline{3-5} \cline{4-5} \cline{5-5} \cline{9-13} \cline{10-13} \cline{11-13} \cline{12-13} \cline{13-13} 
\multirow{2}{*}{{\scriptsize{}Matched Worker}} & \multirow{2}{*}{{\scriptsize{}$\left(20,20\right)$}} & \multirow{2}{*}{{\scriptsize{}$\left(10,20\right)$}} & \multirow{2}{*}{{\scriptsize{}$\left(\boldsymbol{20,10}\right)$}} & \multirow{2}{*}{{\scriptsize{}$\left(\boldsymbol{10,10}\right)$}} &  &  &  & \multirow{2}{*}{{\scriptsize{}Matched Worker}} & \multirow{2}{*}{{\scriptsize{}$\left(\boldsymbol{10,10}\right)$}} & \multirow{2}{*}{{\scriptsize{}$\left(10,20\right)$}} & \multirow{2}{*}{{\scriptsize{}$\left(20,10\right)$}} & \multirow{2}{*}{{\scriptsize{}$\left(\boldsymbol{20,20}\right)$}} & \tabularnewline
 &  &  &  &  &  &  &  &  &  &  &  &  & \tabularnewline
\cline{1-5} \cline{2-5} \cline{3-5} \cline{4-5} \cline{5-5} \cline{9-13} \cline{10-13} \cline{11-13} \cline{12-13} \cline{13-13} 
\multicolumn{5}{c}{} &  &  &  & \multicolumn{5}{c}{} & \tabularnewline
\multicolumn{5}{c}{Swap-2} &  &  &  & \multicolumn{5}{c}{Swap-3} & \tabularnewline
\end{tabular}
\par\end{center}%
\end{minipage}
\end{table}
\par\end{center}

\noindent These examples demonstrate that global, within-group and
weak sorting concepts are not adequate to obtain a fine characterization
of the set of optimal matching distributions. In Section \ref{sec:1},
I lay out the statistical logic behind Becker's sorting result and
restate it by using the upper-set properties of the supermodular order.
By devising multidimensional generalizations of supermodularity, supermodular
order, and concordance, I characterize Pareto improving swaps for
a large set of complementarity structures that allow for negative
complementarities between some skills along with positive complementarities
between some other skills. 

Although these results can be applied to wide-ranging matching markets,
they lack predictive power: the set of optimal matching distributions
may not be a singleton. Lindenlaub (2017) offers a multidimensional
sorting theory with higher predictive power. She adopts three key
assumptions: (a) the agents on each side of the market have the same
number of outcome-relevant attributes, (b) each attribute complements
one and only one attribute on the other side of the market, and (c)
the matching output exhibits either strictly positive complementarity
in all attributes or strictly negative complementarity in all attributes.
More specifically, she considers matching output functions that have
the following form: $Q\left(x,y\right)=Q_{c}\left(x_{c},y_{c}\right)+Q_{s}\left(x_{s},y_{s}\right)$
where the cross-partial derivative of $Q_{i}$ is either strictly
positive for all $i\in\left\{ c,s\right\} $ or strictly negative
for all $i\in\left\{ c,s\right\} $. 

For this framework, she shows that the optimal matching is unique.
In addition, she proves that the optimal matching is a smooth function
under additional restrictions: the agents are distributed with infinitely
many times continuously differentiable probability distribution functions;
$Q_{i}$ is four times continuously differentiable; and $Q_{x_{i},y_{i}}$
is supermodular and log-supermodular. Under these additional restrictions,
she establishes that $\partial y_{i}^{*}/\partial x_{i}>0$ for all
$i\in\left\{ c,s\right\} $, where $y^{*}\coloneqq m^{*}\left(x\right)$
denotes the matched worker of type-\emph{x} firm under optimal matching
function $m^{*}:\mathbb{R}^{2}\rightarrow\mathbb{R}^{2}$. 

However, these assumptions on the output function may not be suitable
for many applications. As an example, consider the following matching
output: $Q\left(x,y\right)=\alpha x_{c}y_{c}+\beta x_{s}y_{s}+\gamma x_{s}y_{c}$.
The results presented by Lindenlaub (2017) apply only if $\alpha\beta>0$
and $\gamma=0$. In this case, once one assumes strictly positive
complementarity between cognitive skills, one must also assume strictly
positive complementarity between social skills, and vice versa ($\alpha\beta>0$).
Also, each attribute of one side can interact with one and only one
attribute of the other side of the market ($\gamma=0$). Dupuy and
Galichon (2014, Table 3) provide empirical evidence suggesting that
(a) positive and negative complementarities between different attributes
exist and (b) some attributes simultaneously complement multiple attributes
in the Dutch marriage market. 

One other drawback of the sorting theorems presented in Proposition
\ref{prop:1} is that they predict extreme sorting patterns as optimal.
However, observing extreme sorting in practice is improbable. A potential
reason\footnote{Search and informational frictions may also cause deviations from
extreme sorting.} is that we do not observe every outcome-relevant attribute. For the
previous example, assume that (a) firms and workers observe each other's
cognitive and social skills, and (b) econometricians observe firms'
and workers' cognitive skills but do not observe their social skills.
In this case, econometricians observe a sorting pattern between firms'
and their workers' cognitive skills aggregated over social skills.
The aggregated matching pattern between cognitive skills may exhibit
mismatch, i.e. deviations from extreme sorting. The framework presented
in Section \ref{sec:1} allows me to explore this source of mismatch.
Consider four types of firms (\emph{x}) and four types of workers
(\emph{y}): $\left\{ \left(10,10\right),\left(10,20\right),\left(20,10\right),\left(20,20\right)\right\} $.
Assume that for skill vectors $\left(10,10\right)$ and $\left(20,20\right)$,
there are four firms. In addition, suppose that there is one firm
for each skill vector $\left(10,20\right)$ and $\left(20,10\right)$.
Similarly, let there be four workers for each skill vector $\left(10,20\right)$
and $\left(20,10\right)$; and one worker for each skill vector $\left(10,10\right)$
and $\left(20,20\right)$. The unique optimal matching scheme between
these firms and workers is given below for the following output function:
$Q\left(x,y\right)=x_{c}y_{c}+2x_{s}y_{s}$.
\noindent \begin{center}
\begin{tabular}{c|c|c|c|c|c|}
\multicolumn{1}{c}{} &  & \multicolumn{4}{c|}{Workers}\tabularnewline
\cline{3-6} \cline{4-6} \cline{5-6} \cline{6-6} 
\multicolumn{1}{c}{} &  & $\left(10,10\right)$ & $\left(10,20\right)$ & $\left(20,10\right)$ & $\left(20,20\right)$\tabularnewline
\hline 
\multirow{4}{*}{Firms} & $\left(10,10\right)$ & 1 & 0 & 3 & 0\tabularnewline
\cline{2-6} \cline{3-6} \cline{4-6} \cline{5-6} \cline{6-6} 
 & $\left(10,20\right)$ & 0 & 1 & 0 & 0\tabularnewline
\cline{2-6} \cline{3-6} \cline{4-6} \cline{5-6} \cline{6-6} 
 & $\left(20,10\right)$ & 0 & 0 & 1 & 0\tabularnewline
\cline{2-6} \cline{3-6} \cline{4-6} \cline{5-6} \cline{6-6} 
 & $\left(20,20\right)$ & 0 & 3 & 0 & 1\tabularnewline
\hline 
\multicolumn{6}{c}{Optimal matching counts between firms and workers}\tabularnewline
\multicolumn{6}{c}{}\tabularnewline
\end{tabular}\hfill{}%
\begin{tabular}{ccccc|c|c|ccc}
 &  &  &  & \multicolumn{1}{c}{} & \multicolumn{1}{c}{} & \multicolumn{1}{c}{} &  &  & \tabularnewline
 &  &  &  & \multicolumn{1}{c}{} & \multicolumn{2}{c}{} &  &  & \tabularnewline
 &  &  &  &  & \multicolumn{2}{c|}{Workers} &  &  & \tabularnewline
\cline{6-7} \cline{7-7} 
\multirow{2}{*}{} &  &  &  &  & 10 & 20 &  &  & \tabularnewline
\cline{3-7} \cline{4-7} \cline{5-7} \cline{6-7} \cline{7-7} 
 &  & \multirow{2}{*}{Firms} & \multicolumn{1}{c|}{} & 10 & 2 & 3 &  &  & \tabularnewline
\cline{5-7} \cline{6-7} \cline{7-7} 
 &  &  & \multicolumn{1}{c|}{} & 20 & 3 & 2 &  &  & \tabularnewline
\cline{3-7} \cline{4-7} \cline{5-7} \cline{6-7} \cline{7-7} 
 &  & \multicolumn{8}{l}{Aggregated matching counts}\tabularnewline
 &  & \multicolumn{8}{l}{by cognitive skills}\tabularnewline
\end{tabular}
\par\end{center}

\noindent In line with Proposition \ref{prop:1}, the optimal matching
scheme (left panel) satisfies weak positive sorting. However, the
aggregated matching between firms' and their workers' cognitive skills
(right panel) exhibits neither positive nor negative sorting.

In Section \ref{sec:2}, I examine a framework introduced by Choo
and Siow (2006) that allows agents to have unobserved characteristics
that are outcome-relevant. In this context, obtaining a tractable
aggregation over unobserved characteristics requires additional assumptions
on the distributions of the unobserved characteristics. In addition,
complementarities between unobserved characteristics cannot be allowed
as they affect the aggregated match between observed attributes in
a way that cannot be controlled by observed attributes. Under these
additional restrictions, I obtain empirically robust sorting results
that (a) suggest milder sorting patterns between observed attributes
and (b) allow us to infer the underlying complementarities between
the observed attributes given empirical matching distribution. 

In this context, the changes in complementarities affect the matching
outcome in a sophisticated way. In Section \ref{sec:2}, I propose
a non-parametric notion of increasing complementarities and establish
comparative static results regarding the changes in complementarities
without making assumptions on either the matching output function
or the distributions of the observable attributes.

To demonstrate an application of the empirically robust sorting results,
I examine the association between agents' health status and their
spouses' education levels among U.S. households in Section \ref{sec:3}
by using the IPUMS-CPS data series for 2010-2017. In literature, many
studies report a positive association between agents' health status
and their spouses' education levels. From the actuarial point of view,
decomposing this association is essential. If one's spouse's education
level is a strong predictor of one's health status, then insurance
companies can reduce the risk that they carry by taking one's spouse's
education level into account. I show that one's spouse's education
level is not a strong predictor of one's health status and identify
an attraction channel which explains the positive association between
agents' health status and their spouses' education levels. It is estimated
that the association is a product of three factors: (a) an attraction
between better-educated individuals, (b) an attraction between healthier
individuals, and (c) a positive association between agents' health
status and their own education levels. This decomposition implies
that the insurers' risk associated with a two-person family plan is
higher than the aggregate risk associated with two individual policies. 

\section*{Empirical Support for Multidimensional Types}

In this section, I survey the recent empirical studies that support
the presence of multiple outcome-relevant attributes in the healthcare,
labor, and marriage/dating markets.

In the healthcare market, the relationship between doctors and patients
is known to be multidimensional. Jagosh et al. (2011) show that effective
communication enhances patient recovery. The authors argue that three
listening skills of health professionals (listening as an essential
component of clinical data gathering and diagnosis; listening as a
healing and therapeutic agent; and listening as a means of fostering
and strengthening the doctor\textendash patient relationship) are
central to successful clinical outcomes. Stavropoulou (2011) indicates
that six aspects of the relationship between doctors and patients
affect nonadherence to medication using the European Social Survey.
Nonadherence to medication is also proven to be a complex and multidimensional
healthcare problem by Hugtenburg et al. (2013). Similarly, Mazzi et
al. (2018) identify four attributes of doctors and three characteristics
of patients that are essential to successful clinical outcomes by
using an integrated survey of thirty-one European countries. Belasen
and Belasen (2018) provide evidence suggesting that different aspects
of the relationships between doctors and patients affect not only
the clinical outcome but also patients' rankings of hospitals. 

The empirical findings in labor economics literature also suggest
that workers and firms have multidimensional types. Deming (2017)
shows that workers' cognitive and social skills are important determinants
of wages in the U.S. labor market. Girsberger et al. (2018) add manual
skill to that list for the Swiss labor market by using data from the
Social Protection and Labour Market (SESAM) panel. Guvenen et al.
(2018) analyze the skill mismatch between workers and firms in the
U.S. labor market. Their analysis indicates that verbal and math skills
have statistically significant effects on workers' wages.

Hitsch et al. (2010) study mating behavior in the U.S. dating market
by using a large dataset provided by an online dating website. They
show that the differences in age, educational attainment, and body
mass index decrease the probability of dating. Belot and Francesconi
(2012) confirm these findings by studying the speed dating patterns
of individuals based on a British dataset. They find that physical
attractiveness factors (age, height, and body mass index) play an
essential role in the earlier stages of the relationship. Klofstad
et al. (2013) add political views to that list. By analyzing a large
dataset provided by another online dating website, the authors show
that couples tend to share the same political preferences.

Gemici and Laufer (2010) study cohabitation, marriage and separation
patterns in the U.S. mating market. They report that age, educational
attainment, and race are key variables in explaining agents' choices.
Dupuy and Galichon (2014) add other important variables to that list
by studying the Dutch marriage market. They show that personality
traits, such as emotional stability and conscientiousness, are also
important determinants of the Dutch household formation. Domingue
et al. (2014) analyze the genetic similarities between married couples
by using the Health and Retirement Study and information from 1.7
million single-nucleotide polymorphisms. Their results demonstrate
that similar genetic types attract each other. They also find that
educational similarities between spouses are stronger than genetic
similarities. The thorough examination of the household formation
allows researchers to explain changes in the household income inequality.
For example, Greenwood et al. (2014) argue that similarities between
spouses' education levels contributed to the increasing household
income inequality in the U.S. between 1960 and 2005. 

\section{\label{sec:1}The Matching Model with Multidimensional Types}

In this section, I study a general model of two-sided matching markets.
The only difference between the framework examined by Becker (1973)
and the one presented in this section is that I allow the agents to
have multiple outcome-relevant attributes while Becker (1973) assumes
that each agent has only one outcome-relevant attribute. I start by
outlining the general framework in detail. 

\textit{Agents:} There are two sides of the matching market, namely
firms and workers. A generic firm is denoted by $x$, and a generic
worker is denoted by $y$. Every firm has $K$ productive attributes,
i.e. $x\in\mathbb{R}^{K}$, and each worker has $L$ productive attributes,
i.e. $y\in\mathbb{R}^{L}$. It is assumed that the overall measures
of firms and workers coincide. The firms and workers are distributed
according to cumulative distribution functions $F:\mathbb{R}^{K}\rightarrow\left[0,1\right]$
and $G:\mathbb{R}^{L}\rightarrow\left[0,1\right]$, respectively. 

\textit{Matching Distribution:} Matching distribution $M:\mathbb{R}^{K}\times\mathbb{R}^{L}\rightarrow\left[0,1\right]$
is a cumulative distribution function associated with a particular
matching scheme. More specifically, $M\left(x,y\right)$ represents
the fraction of matched firm-worker couples with attributes less than
or equal to $\left(x,y\right)$. 

\noindent Given $F$ and $G$, matching distribution $M$ satisfies
no-single property if and only if 
\noindent \begin{center}
(a) $\underset{y\rightarrow\infty^{L}}{\lim}M\left(x,y\right)=F\left(x\right)$
for all $x\in\mathbb{R}^{K}$, and (b) $\underset{x\rightarrow\infty^{K}}{\lim}M\left(x,y\right)=G\left(y\right)$
for all $y\in\mathbb{R}^{L}$. 
\par\end{center}

\noindent Let $\mathcal{M}\left(F,G\right)$ denote the set of matching
distributions that satisfy no-single property given $F$ and $G$.

\textit{Output Function:} A match between a firm and a worker with
attributes $x$ and $y$ generates a matching output. The matching
output is determined by exogenously specified output function $Q:\mathbb{R}^{K}\times\mathbb{R}^{L}\rightarrow\mathbb{R}_{++}$,
and denoted by $Q\left(x,y\right)$. For any unmatched agent, the
output is assumed to be $0$. 

\textit{Planner's Problem:} Given $F$, $G$, and $Q$, the social
planner maximizes the aggregate output by choosing a matching distribution
that satisfies no-single property: 
\noindent \begin{center}
$\underset{M\in\mathcal{M}\left(F,G\right)}{\max}\int QdM.$
\par\end{center}

Two key concepts of the matching theory are complementarity and sorting.
In this context, \emph{positive(negative) complementarity} between
a firm's $i^{th}$ attribute and its worker's $j^{th}$ attribute
means that the marginal product of the firm's $i^{th}$ attribute
is increasing(decreasing) in its worker's $j^{th}$ attribute. In
unidimensional matching literature, complementarities are formulated
by using supermodularity. Function $Q:\mathbb{R}^{2}\rightarrow\mathbb{R}$
is called \emph{(strictly) supermodular} if for all $x^{\prime}>x$
and $y^{\prime}>y$, it holds that

\noindent 
\[
\left\{ Q\left(x^{\prime},y^{\prime}\right)-Q\left(x,y^{\prime}\right)\right\} -\left\{ Q\left(x^{\prime},y\right)-Q\left(x,y\right)\right\} \geq\left(>\right)0.
\]

\noindent Similarly, $Q$ is \emph{(strictly) submodular} if $-Q$
is (strictly) supermodular; and $Q$ satisfies \emph{modularity} if
$Q$ is both supermodular and submodular. In the multidimensional
setting, supermodularity can also be used with a slight modification
to formulate multidimensional complementarities. Towards this end,
I introduce \emph{i,j pairwise supermodularity}.
\begin{defn}
\label{def:1}Function $Q:\mathbb{R}^{K}\times\mathbb{R}^{L}\rightarrow\mathbb{R}_{++}$
is 

\noindent \textbf{(a)} (strictly) i,j pairwise supermodular if $Q\left(x_{i},x_{-i},y_{j},y_{-j}\right)$
is a (strictly) supermodular function of  $\left(x_{i},y_{j}\right)$
for all $\left(x_{-i},y_{-j}\right)\in\mathbb{R}^{K-1}\times\mathbb{R}^{L-1}$;

\noindent \textbf{(b)} (strictly) i,j pairwise submodular if $-Q$
is (strictly) i,j pairwise supermodular; and

\noindent \textbf{(c) }i,j pairwise modular if $Q$ is both i,j pairwise
supermodular and i,j pairwise submodular.
\end{defn}
By construction, the pairwise modularity concepts can be used to formulate
the relationships between the marginal product of a firm's $i^{th}$
attribute and its worker's $j^{th}$ attribute. This aspect is easy
to observe when the output function is smooth: $Q:\mathbb{R}^{K}\times\mathbb{R}^{L}\rightarrow\mathbb{R}$
satisfies i,j pairwise supermodularity if and only if its cross-partial
derivative with respect to $x_{i}$ and $y_{j}$ is positive. Through
the use of pairwise modularity concepts, I define two multidimensional
complementarity classes. Consider two disjoint subsets of $\left\{ 1,\ldots,K\right\} \times\left\{ 1,\ldots,L\right\} $:
P and N.
\begin{defn}
\label{def:2}Function $Q:\mathbb{R}^{K}\times\mathbb{R}^{L}\rightarrow\mathbb{R}_{++}$
exhibits (strict) P,N modularity if and only if $Q\left(x,y\right)$
is 

\noindent \textbf{(a)} (strictly) i,j pairwise supermodular for all
$\left(i,j\right)\in P$; 

\noindent \textbf{(b)} (strictly) p,q pairwise submodular for all
$\left(p,q\right)\in N$; and 

\noindent \textbf{(c)} m,n pairwise modular for all $\left(m,n\right)\notin P\cup N$. 

\noindent Let $\mathbb{C}\left(P,N\right)$ and \emph{$\mathbb{C}_{+}\left(P,N\right)$
}denote the sets of functions which satisfy P,N modularity and strict
P,N modularity for set parameters P and N. 
\end{defn}
Here $P$ represents the set of firm-worker attribute pairs for which
the output function exhibits positive complementarity. Similarly,
$N$ represents the set of firm-worker attribute pairs for which the
output function exhibits negative complementarity. The strongest element
of these complementarity classes is that they do not allow for complementarity
between a firm's $i^{th}$ attribute and its worker's $j^{th}$ attribute
to change signs (from positive to negative). For example, an output
function which exhibits (a) strictly positive complementarity between
a firm's and its worker's cognitive skills for some levels of its
worker's social skill, and (b) strictly negative complementarity between
the firm's and its worker's cognitive skills for some levels of its
worker's social skill cannot be formulated by using a P,N modular
function. Despite this shortcoming, P,N modular functions can capture
various output function forms which have been frequently used in the
matching literature.

\noindent 
\begin{table}[H]
\noindent \centering{}\caption{\label{tab:1}Some examples of P,N modular functions frequently used
in matching literature}
{\scriptsize{}}%
\begin{tabular}{c|c|cc}
\multirow{2}{*}{\textbf{\scriptsize{}Output function}} & \multirow{2}{*}{\textbf{\scriptsize{}Strict P,N modularity}} & \multirow{2}{*}{\textbf{\scriptsize{}Reference}} & \tabularnewline
 &  &  & \tabularnewline
\cline{1-3} \cline{2-3} \cline{3-3} 
\multirow{2}{*}{{\scriptsize{}$Q\left(x,y\right)=x^{\prime}Ay=\underset{i=1}{\overset{K}{\sum}}\underset{j=1}{\overset{L}{\sum}}a_{ij}x_{i}y_{j}$ }} & \multirow{2}{*}{{\scriptsize{}$P=\left\{ \left(i,j\right):a_{ij}>0\right\} $ and
$N=\left\{ \left(i,j\right):a_{ij}<0\right\} $}} & \multirow{2}{*}{{\scriptsize{}Dupuy and Galichon (2014)}} & \tabularnewline
 &  &  & \tabularnewline
\cline{1-3} \cline{2-3} \cline{3-3} 
\multirow{2}{*}{{\scriptsize{}$Q\left(x,y\right)=\underset{i=1}{\overset{K}{\sum}}Q_{i}\left(x_{i},y_{i}\right)$ }} & \multirow{2}{*}{{\scriptsize{}$P=\left\{ 1,...,K\right\} $ and $N=\emptyset$ if
$\partial^{2}Q_{i}\left(x_{i},y_{i}\right)/\partial x_{i}\partial y_{i}>0$}} & \multirow{2}{*}{{\scriptsize{}Lindenlaub (2017)}} & \tabularnewline
 &  &  & \tabularnewline
\cline{1-3} \cline{2-3} \cline{3-3} 
\end{tabular}
\end{table}

Due to the simplicity of their interpretation and their frequent use,
understanding the optimal sorting patterns for P,N modular output
functions is of theoretical and empirical interest. To achieve this
goal, I introduce three multidimensional sorting patterns. The first
sorting pattern is a direct extension of the global sorting pattern
presented by Becker (1973).
\begin{defn}
\label{def:3}Matching distribution $M:\mathbb{R}^{K}\times\mathbb{R}^{L}\rightarrow\left[0,1\right]$
satisfies

\noindent \textbf{(a)} positive sorting between firms' $i^{th}$ and
their workers' $j^{th}$ attributes if and only if 
\noindent \begin{center}
$\left(x_{i}^{\prime}-x_{i}\right)\left(y_{j}^{\prime}-y_{j}\right)\geq0$
for all $\left(x,y\right),\left(x^{\prime},y^{\prime}\right)\in supp\left(M\right)$,
and 
\par\end{center}

\noindent \textbf{(b)} negative sorting between firms' $i^{th}$ and
their worker's $j^{th}$ attributes if and only if 
\noindent \begin{center}
$\left(x_{i}^{\prime}-x_{i}\right)\left(y_{j}^{\prime}-y_{j}\right)\leq0$
for all $\left(x,y\right),\left(x^{\prime},y^{\prime}\right)\in supp\left(M\right)$. 
\par\end{center}

\end{defn}
Similar to P,N modularity, I combine pairwise sorting patterns to
construct a multidimensional sorting class. 
\begin{defn}[Global P,N sorting]
\label{def:4}Matching distribution $M:\mathbb{R}^{K}\times\mathbb{R}^{L}\rightarrow\left[0,1\right]$
satisfies global P,N sorting if and only if the following conditions
hold for all $\left(x,y\right),\left(x^{\prime},y^{\prime}\right)\in supp\left(M\right)$,

\noindent \textbf{(a)} $\left(x_{i}^{\prime}-x_{i}\right)\left(y_{j}^{\prime}-y_{j}\right)\geq0$
for all $\left(i,j\right)\in P$, and 

\noindent \textbf{(b)} $\left(x_{p}^{\prime}-x_{p}\right)\left(y_{q}^{\prime}-y_{q}\right)\leq0$
for all $\left(p,q\right)\in N$.

\noindent In other words, a matching distribution exhibits global
P,N sorting if and only if it exhibits (a) positive sorting between
firms' $i^{th}$ and their workers' $j^{th}$ attributes for all $\left(i,j\right)\in P$
and (b) negative sorting between firms' $p^{th}$ and their workers'
$q^{th}$ attributes for all $\left(p,q\right)\in N$. 
\end{defn}
Although global P,N sorting is a clear sorting pattern between two
multidimensional agents, it requires simultaneous sorting between
different attributes. Consequently, its existence is tied to the distributions
of the agents. 
\begin{example}
\label{ex:1}Consider a labor market with equal numbers of two types
of firms: $\left(10,10\right)$ and $\left(20,20\right)$; and equal
numbers of two types of workers: $\left(10,20\right)$ and $\left(20,10\right)$.
Notice that matching a $\left(10,20\right)$ worker with a $\left(10,10\right)$
firm and a $\left(20,10\right)$ worker with a $\left(20,20\right)$
firm violates positive sorting between the second attributes. Similarly,
the swap between these two couples, i.e. matching a $\left(10,20\right)$
worker with a $\left(20,20\right)$ firm and a $\left(20,10\right)$
worker with a $\left(10,10\right)$ firm, violates positive sorting
between the first attributes. Therefore, it is not possible to observe
simultaneous positive sorting between the first attributes and the
second attributes without inefficiently leaving some agents unmatched.
\end{example}
Due to the existence issue, I study two alternative sorting patterns
that exist for arbitrary $F,G,P,$ and $N$. The next sorting pattern
is inspired by Chiappori et al. (2017). The authors analyze a matching
model in which each side of the market is represented by one continuous
variable (socioeconomic status) and one binary variable (smoking habit).
They categorize couples into two main groups. In the first group,
both men and women are non-smokers. In the second group, at least
one spouse is a smoker. They assume that the matching output of a
couple is the multiplication of spouses' socioeconomic status. If
there is a smoker in the household, then the output is scaled down
by a constant. Under this complementarity structure, they predict
positive sorting between agents' socioeconomic status \emph{within}
each group. Although this sorting result immediately follows from
Becker's (1973) unidimensional sorting theory, the idea of within-group
sorting can be deployed in a more general setting.
\begin{defn}
\label{def:5}Matching distribution $M:\mathbb{R}^{K}\times\mathbb{R}^{L}\rightarrow\left[0,1\right]$
satisfies 

\noindent \textbf{(a) }within-group positive sorting between firms'
$i^{th}$ and their workers' $j^{th}$ attributes if and only if for
all $\left(x,y\right),\left(x^{\prime},y^{\prime}\right)\in supp\left(M\right)$
such that $\left(\clubsuit\right)$ $x_{k}=x_{k}^{\prime}$ for all
$k\neq i$ and $\left(\spadesuit\right)$ $y_{l}=y_{l}^{\prime}$
for all $l\neq j$, it holds that
\noindent \begin{center}
$\left(x_{i}^{\prime}-x_{i}\right)\left(y_{j}^{\prime}-y_{j}\right)\geq0$;
and 
\par\end{center}

\noindent \textbf{(b) }within-group negative sorting between firms'
$i^{th}$ and their workers' $j^{th}$ attributes if and only if for
all $\left(x,y\right),\left(x^{\prime},y^{\prime}\right)\in supp\left(M\right)$
such that $\left(\clubsuit\right)$ $x_{k}=x_{k}^{\prime}$ for all
$k\neq i$ and $\left(\spadesuit\right)$ $y_{l}=y_{l}^{\prime}$
for all $l\neq j$, it holds that
\noindent \begin{center}
$\left(x_{i}^{\prime}-x_{i}\right)\left(y_{j}^{\prime}-y_{j}\right)\leq0$. 
\par\end{center}

\end{defn}
Similar to P,N sorting, I combine within-group sorting patterns to
define a within-group sorting class.
\begin{defn}[Within-group P,N Sorting]
\label{def:6}Matching distribution $M:\mathbb{R}^{K}\times\mathbb{R}^{L}\rightarrow\left[0,1\right]$
satisfies within-group P,N sorting if and only if it satisfies

\noindent \textbf{(a) }within-group positive sorting between firms'
$i^{th}$ and their workers' $j^{th}$ attributes for all $\left(i,j\right)\in P$;
and 

\noindent \textbf{(b) }within-group negative sorting between firms'
$p^{th}$ and their workers' $q^{th}$ attributes for all $\left(p,q\right)\in N$. 
\end{defn}
Within-group P,N sorting is not a very informative sorting pattern
as some of the nonoptimal matching distributions may satisfy within-group
P,N sorting.
\begin{example}
\label{ex:2}Consider a matching market with equal numbers of two
types of workers: $\left(10,20\right)$ and $\left(20,10\right)$;
and equal numbers of two types of firms: $\left(10,20\right)$ and
$\left(20,10\right)$. Notice that matching $\left(10,20\right)$
workers with $\left(20,10\right)$ firms, and $\left(20,10\right)$
workers with $\left(10,20\right)$ firms, satisfies within-group P,N
sorting for

\noindent $P=\left\{ \left(1,1\right),\left(2,2\right)\right\} $
and $N=\left\{ \right\} $. However, it is clear that a swap of the
partners between firm-worker couples $\left\{ \left(10,20\right),\left(20,10\right)\right\} $
and $\left\{ \left(20,10\right),\left(10,20\right)\right\} $ strictly
improves the aggregate output for any strictly P,N modular output
function. For example, matching same types with each other is associated
with strictly higher aggregate output for output function $Q\left(x,y\right)=x_{1}y_{1}+x_{2}y_{2}$: 
\noindent \begin{center}
$100+400+400+100=1000>800=200+200+200+200.$
\par\end{center}

\noindent The right-hand side of the equation above equals the total
output produced by firm-worker couples $\left\{ \left(10,20\right),\left(20,10\right)\right\} $
and $\left\{ \left(20,10\right),\left(10,20\right)\right\} $; and
the left-hand side of the equation is the total output produced by
firm-worker couples $\left\{ \left(10,20\right),\left(10,20\right)\right\} $
and $\left\{ \left(20,10\right),\left(20,10\right)\right\} $.
\end{example}
As it is demonstrated in Example \ref{ex:2}, a fine characterization
of the set of optimal matching distributions cannot be obtained by
using within-group P,N sorting. As such, I adopt a more constructive
approach to obtain a fine description of the set of optimal matching
distributions. To motivate the idea behind this approach, I first
lay out a statistical course to obtain Becker's (1973) sorting results. 
\begin{thm}[Becker's (1973) sorting results]

\noindent In a unidimensional matching market, i.e. $K=L=1$, 

\noindent \textbf{(a)} positive sorting is an optimal matching distribution
when the output function is supermodular;

\noindent \textbf{(b)} negative sorting is an optimal matching distribution
when the output function is submodular;

\noindent \textbf{(c)} positive sorting is the unique optimal matching
distribution when the output function is strictly supermodular; and

\noindent \textbf{(d)} negative sorting is the unique optimal matching
distribution when the output function is strictly submodular.
\end{thm}
It is easy to establish these sorting results by using the upper-set
properties in the supermodular order. Distribution $M:\mathbb{R}\times\mathbb{R}\rightarrow\left[0,1\right]$
dominates $M^{\prime}:\mathbb{R}\times\mathbb{R}\rightarrow\left[0,1\right]$
in the \emph{supermodular order }if and only if  $\int QdM\geq\int QdM^{\prime}$
for all supermodular $Q:\mathbb{R}\times\mathbb{R}\rightarrow\mathbb{R}$.
Muller and Scarsini (2010), and Meyer and Strulovici (2013) offer
an alternative characterization of the supermodular order. 

\begin{defn}
\label{def:7}\ 

\noindent \textbf{(a) }A pair of couples $\left(x,y\right),\left(x^{\prime},y^{\prime}\right)\in\mathbb{R}^{2}$
is (weakly) concordant if and only if $\left(x-x^{\prime}\right)\left(y-y^{\prime}\right)>\left(\geq\right)0$. 

\noindent \textbf{(b) }A pair of couples $\left(x,y\right),\left(x^{\prime},y^{\prime}\right)\in\mathbb{R}^{2}$
is (weakly) discordant if and only if $\left(x-x^{\prime}\right)\left(y-y^{\prime}\right)<\left(\leq\right)0$. 

\noindent \textbf{(c) }A concordance improving transfer is a uniform
probability transfer from a discordant bivariate pair to the concordant
bivariate pair that is obtained via the swap between the discordant
bivariate pair. Let $\tau\left(x,y,x^{\prime},y^{\prime};\alpha\right)$
denote the concordance improving transfer with $\alpha\geq0$ weight
which increases densities by $\alpha$ at $\left(x,y\right)$ and
$\left(x^{\prime},y^{\prime}\right)$ such that $\left(x-x^{\prime}\right)\left(y-y^{\prime}\right)>0$;
and decreases densities at $\left(x,y^{\prime}\right)$ and $\left(x^{\prime},y\right)$
by $\alpha$. 
\end{defn}
\begin{thm}[Muller and Scarsini (2010), and Meyer and Strulovici (2013)]
 \emph{$M:\mathbb{R}\times\mathbb{R}\rightarrow\left[0,1\right]$
}dominates $M^{\prime}:\mathbb{R}\times\mathbb{R}\rightarrow\left[0,1\right]$
in the supermodular order if and only if $M$ can be obtained from
$M^{\prime}$ via concordance improving transfers. 
\end{thm}
Based on this alternative characterization, it is easy to see that
positive sorting is the \emph{dominant} matching distribution in the
supermodular order. Similarly, negative sorting is strictly \emph{dominated}
by any other matching distribution in the supermodular order. Consequently,
Becker's (1973) sorting result is obtained. Following a similar logic,
a fine characterization of the set of optimal matching distributions
for the multidimensional setting can be established by identifying
the changes in the matching distribution that increase the aggregate
output when the matching output is P,N modular. On this note, I define
P,N modular order in Definition \ref{def:8}.
\begin{defn}
\label{def:8}Distribution \emph{$M:\mathbb{R}^{K}\times\mathbb{R}^{L}\rightarrow\left[0,1\right]$
}dominates $M^{\prime}:\mathbb{R}^{K}\times\mathbb{R}^{L}\rightarrow\left[0,1\right]$
in P,N modular order, denoted by $M\succeq_{P,N}M^{\prime}$, if and
only if $\int QdM\geq\int QdM^{\prime}$ for all $Q\in\mathbb{C}\left(P,N\right)$.
Distribution $M$ strictly dominates $M^{\prime}$ in P,N modular
order if and only if (a) $M$ dominates $M^{\prime}$ in P,N modular
order; and (b) $M$ is not dominated by $M^{\prime}$ in P,N modular
order.
\end{defn}
\noindent Next, I formally define the sets of P,N dominant and P,N
undominated distributions that are essential to characterizing the
set of optimal matching distributions.
\begin{defn}
\label{def:9}Matching distribution $M\in\mathcal{M}\left(F,G\right)$
is P,N dominant if it dominates every $M^{\prime}\in\mathcal{M}\left(F,G\right)$
in P,N modular order. Similarly, $M\in\mathcal{M}\left(F,G\right)$
is P,N undominated if there does not exist $M^{\prime}\in\mathcal{M}\left(F,G\right)$
which strictly dominates $M$ in P,N modular order. 
\end{defn}
The next theorem restates Becker's (1973) sorting result by using
multidimensional concepts. By doing so, it makes the theoretical differences
between unidimensional and multidimensional matching markets clear. 
\begin{thm}[Unidimensional Sorting - Multidimensional Concepts]
\label{thm:3}Let $K=L=1$ and $P\cup N=\left\{ \left(1,1\right)\right\} $.

\noindent a) For any P and N, a matching distribution is P,N\textbf{
}dominant if and only if it is P,N undominated.

\noindent b) For any P and N, there is only one P,N dominant distribution.

\noindent c) For any $Q\in\mathbb{C}\left(P,N\right)$, the P,N dominant
distribution solves the planner's problem.

\noindent d) For any $Q\in\mathbb{C}_{+}\left(P,N\right)$, the P,N
dominant distribution is the unique solution to the planner's problem.

\noindent e) The P,N\textbf{ }dominant matching distribution is positive
assortative matching when $P=\left\{ \left(1,1\right)\right\} $ and
$N=\left\{ \right\} $\emph{: }
\noindent \begin{center}
$\Lambda\left(x,y\right)=\min\left\{ F\left(x\right),G\left(y\right)\right\} .$
\par\end{center}

\noindent f) The P,N dominant matching distribution is negative assortative
matching when $P=\left\{ \right\} $ and $N=\left\{ \left(1,1\right)\right\} $\emph{:}
\noindent \begin{center}
$\varOmega\left(x,y\right)=\max\left\{ F\left(x\right)+G\left(y\right)-1,0\right\} .$
\par\end{center}

\end{thm}
There are two key aspects of optimality in matching markets with unidimensional
agents. First of all, the set of P,N undominated distributions is
singleton when $P\cup N\neq\emptyset$. Secondly, there exists a P,N
dominant distribution. These two features make it possible to fully
characterize the solution to the planner's problem for strictly P,N
modular output functions in unidimensional case. However, they do
not apply to the multidimensional setting. 
\begin{example}
\label{ex:3}Consider a matching market with equal numbers of two
types of workers: $\left(10,20\right)$ and $\left(20,10\right)$,
equal numbers of two types of firms: $\left(10,10\right)$ and $\left(20,20\right)$,
and a class of matching output functions with parameter $\gamma\in\left[0,1\right]$:
$Q\left(x,y;\gamma\right)=\gamma x_{1}y_{1}+\left(1-\gamma\right)x_{2}y_{2}$.
Note that for all values of $\gamma$, the output function exhibits
P,N modularity for $P=\left\{ \left(1,1\right),\left(2,2\right)\right\} $
and $N=\left\{ \right\} $. 

\emph{The set of P,N undominated distributions is not singleton:}
Consider matching distribution $M$ under which every $\left(10,20\right)$
worker is matched with a $\left(10,10\right)$ firm; and every $\left(20,10\right)$
worker is matched with a $\left(20,20\right)$ firm. It immediately
follows from Becker's sorting result that $M$ is P,N undominated
since it is the unique solution to the planner's problem when $\gamma=1$.
Alternative matching distribution $M^{\prime}$ under which every
$\left(20,10\right)$ worker is matched with a $\left(10,10\right)$
firm, and every $\left(10,20\right)$ worker is matched with a $\left(20,20\right)$
firm is also P,N undominated since it is the unique solution to the
planner's problem when $\gamma=0$. 

\emph{The set of P,N dominant distributions is empty:} Consider a
matching distribution under which some $\left(10,20\right)$ workers
are matched with $\left(10,10\right)$ firms; and some $\left(20,10\right)$
workers are matched with $\left(20,20\right)$ firms. This matching
distribution cannot be P,N dominant as $M^{\prime}$ is associated
with strictly higher aggregate output when $\gamma=0$. Similarly,
a matching distribution under which some $\left(20,10\right)$ workers
are matched with $\left(10,10\right)$ firms, and some $\left(10,20\right)$
workers are matched with $\left(20,20\right)$ firms cannot be P,N
dominant as $M$ is associated with strictly higher aggregate output
when $\gamma=1$. 
\end{example}
Due to the fact that the set of P,N dominant distributions is neither
singleton nor non-empty for arbitrary model parameters, characterization
of the optimal matching distributions is a non-trivial task. I obtain
a fine description of the optimal matching distributions, by characterizing
Pareto improving swaps in Lemma \ref{lem:1} below.
\begin{defn}
\label{def:10}A pair of firm-worker couples $\left(x,y\right),\left(x^{\prime},y^{\prime}\right)\in\mathbb{R}^{K}\times\mathbb{R}^{L}$
is P,N weak concordant if 

\noindent \textbf{(a) }$\left(x_{i}-x_{i}^{\prime}\right)\left(y_{j}-y_{j}^{\prime}\right)\geq0$
for all $\left(i,j\right)\in P$; and

\noindent \textbf{(b) }$\left(x_{p}-x_{p}^{\prime}\right)\left(y_{q}-y_{q}^{\prime}\right)\leq0$
for all $\left(p,q\right)\in N$. 

\noindent A P,N concordant pair is a P,N weak concordant pair such
that $\left(\clubsuit\right)$ some of the inequalities in (a) hold
with strict inequality for some $\left(i,j\right)\in P$; or $\left(\spadesuit\right)$
some of the inequalities in (b) hold with strict inequality for some
$\left(p,q\right)\in N$.
\end{defn}
\begin{defn}
\noindent \label{def:11}A P,N concordance improving transfer is a
uniform probability transfer from an N,P weak concordant pair of couples
to the P,N weak concordant pair of couples that is obtained via the
swap between the N,P weak concordant pair of couples. 
\end{defn}
\begin{lem}
\noindent \label{lem:1}Distribution \emph{$M:\mathbb{R}^{K}\times\mathbb{R}^{L}\rightarrow\left[0,1\right]$
}dominates $M^{\prime}:\mathbb{R}^{K}\times\mathbb{R}^{L}\rightarrow\left[0,1\right]$
in P,N modular order if and only if\emph{ }$M$ can be obtained from
$M^{\prime}$ via a sequence of P,N concordance improving transfers.
\end{lem}
This alternative characterization of dominance in P,N modular order
allows me to obtain a finer description of the set of optimal matching
distributions. Consider matching distribution $M\in\mathcal{M}\left(F,G\right)$
under which there exists an N,P concordant pair of matched couples.
Due to Lemma \ref{lem:1}, it is easy to see that a swap between the
N,P concordant pair of couples (strictly) increases the aggregate
output when the output function is (strictly) P,N modular. Consequently,
a matching distribution under which there exists an N,P concordant
pair of matched couples cannot be obtained as optimal for strictly
P,N modular output functions. This observation rules out the nonoptimal
matching distribution illustrated in Example \ref{ex:2}, and allows
me to define a new sorting pattern. 
\begin{defn}
\label{def:12}Matching distribution $M\in\mathcal{M}\left(F,G\right)$
satisfies weak P,N sorting if there does not exist an N,P concordant
pair of couples with a positive mass under $M$.
\end{defn}
\noindent The next proposition establishes the link between the proposed
sorting patterns and the set of optimal matching distributions for
the proposed complementarity structures.
\begin{prop}[Multidimensional Sorting]
\label{prop:1}\ 

\noindent \textbf{1.} For arbitrary distributions of the agents $F$
and $G$, and two disjoint sets of firm-worker attribute pairs $P$
and $N$,

\noindent \textbf{1.a)} every weak P,N assortative matching distribution
satisfies within-group P,N sorting;

\noindent \textbf{1.b)} for every $Q\in\mathbb{C}_{+}\left(P,N\right)$,
every solution to the planner's problem satisfies weak P,N sorting;
and

\noindent \textbf{1.c)} for every $Q\in\mathbb{C}\left(P,N\right)$,
there exists a weak P,N assortative matching distribution that solves
the planner's problem.

\noindent \textbf{2.} Suppose that there exists $M\in\mathcal{M}\left(F,G\right)$
that satisfies global P,N sorting. Then,

\noindent \textbf{2.a)} for every $Q\in\mathbb{C}_{+}\left(P,N\right)$,
every solution to the planner's problem satisfies global P,N sorting;
and

\noindent \textbf{2.b)} for every $Q\in\mathbb{C}\left(P,N\right)$,
$M$ solves the planner's problem. 
\end{prop}
Proposition \ref{prop:1} summarizes the relationship between the
proposed sorting patterns and multidimensional complementarities.
First, it states that within-group P,N sorting is the least informative
sorting pattern among all. Secondly, it offers a partial identification
of the set of optimal matching distributions: (i) the set of optimal
matching distributions and the set of weak P,N assortative matching
distributions intersect for P,N modular output functions; and (ii)
the set of optimal matching distributions is a subset of the set of
weak P,N assortative matching distributions for strictly P,N modular
output functions. On the other hand, a weak P,N assortative matching
distribution cannot be optimal for any strictly P,N modular output
function when it is strictly dominated by another matching distribution.
This observation immediately follows from the fact that the absence
of N,P concordant pairs is necessary but not sufficient for undominance
in the P,N modular order. Altogether, Proposition \ref{prop:1} suggests
that P,N sorting is a suitable general sorting class to study the
multidimensional matching markets.

Although it provides new insights into optimal matching with multidimensional
agents, this sorting result is not very practical for empirical purposes.
Proposition \ref{prop:1} states that one cannot observe N,P concordant
pairs under strictly P,N modular functions when every outcome-relevant
attribute is observed by econometricians. In practice, econometricians
only observe sorting patterns between observed attributes that are
aggregated over unobserved characteristics. Therefore, linking production
complementarities between observed attributes to matching patterns
between these attributes when outcome-relevant and unobserved characteristics
are present is of theoretical and empirical interest. In the next
section, I examine a matching model in which agents have outcome-relevant
and unobserved characteristics in addition to their observed attributes.

\section{\label{sec:2}Sorting with Unobserved Characteristics}

Choo and Siow (2006) propose a matching model with multidimensional
agents in which the agents have outcome-relevant characteristics that
are not observed by econometricians. In this section, I examine a
homoskedastic extension of their model. 

Consider a two-sided matching model with equal numbers of firms and
workers. Here, firm $f\in\mathbb{F}$ is described by full attribute
vector $\tilde{x}^{f}$, and $\tilde{y}^{w}$ describes worker $w\in\mathbb{W}$.
Matching scheme $\boldsymbol{\tilde{m}}=\left\{ \tilde{m}_{fw}\right\} $
is a matrix such that the cell value associated with firm $f$ and
worker $w$, i.e. $\tilde{m}_{fw}$, equals one if firm $f$ and worker
$w$ are matched, and it equals zero otherwise. Let $x^{f}\in\mathbb{R}^{K}$
denote the observable attributes of firm $f$, and $y^{w}\in\mathbb{R}^{L}$
denote the observable attributes of worker $w$. The matching output
produced by firm $f$ and worker $w$, denoted by $\tilde{Q}\left(\tilde{x}^{f},\tilde{y}^{w}\right)$,
is determined by firm $f$'s and worker $w$'s full attributes. An
optimal matching matrix maximizes the aggregate output. 

Notice that an optimal matching matrix is determined by agents' full
attributes. However, econometricians observe only observable attributes.
Consequently, the empirical goal is to estimate the complementarities
between observable attributes by using the optimal matching density
function between observable attributes implied by an optimal matching
matrix. Optimal matching density function $m:\mathbb{R}^{K}\times\mathbb{R}^{L}\rightarrow\left[0,1\right]$,
associated with matching matrix $\boldsymbol{\tilde{m}}=\left\{ \tilde{m}_{fw}\right\} $,
is defined through aggregation over unobserved characteristics:
\noindent \begin{center}
$m\left(x,y\right)=\dfrac{\underset{f}{\sum}\underset{w}{\sum}\tilde{m}_{fw}1_{\left\{ x^{f}=x\ and\ y^{w}=y\right\} }}{\underset{f}{\sum}\underset{w}{\sum}\tilde{m}_{fw}}.$
\par\end{center}

\noindent Notice that if the effect of unobserved characteristics
on the optimal matching between observables cannot be controlled by
observed attributes, then one cannot consistently estimate the preferences
on observed attributes. Example \ref{ex:4} demonstrates two channels
through which the unobserved characteristics affect optimal matching
between observables, in a way that cannot be explained by observed
attributes when the distributions of unobserved characteristics conditional
on observed attributes are unknown. 
\begin{example}
\label{ex:4}Consider a matching market with two firms: $\left\{ \left(10,10\right),\left(20,20\right)\right\} $,
and two workers: $\left\{ \left(20,10\right),\left(10,20\right)\right\} $.
Suppose that econometricians do not observe the first characteristics,
and the second attributes are observable. 

\emph{Complementarity between unobserved characteristics:} $\tilde{Q}\left(\tilde{x}^{f},\tilde{y}^{w}\right)=5\tilde{x}_{1}^{f}\tilde{y}_{1}^{w}+2\tilde{x}_{2}^{f}\tilde{y}_{2}^{w}$

\emph{Complementarity between unobserved and observed attributes:}
$\tilde{Q}\left(\tilde{x}^{f},\tilde{y}^{w}\right)=2\tilde{x}_{2}^{f}\tilde{y}_{1}^{w}+\tilde{x}_{2}^{f}\tilde{y}_{2}^{w}$

\noindent For these two matching output functions, the following optimal
matching matrix and density are obtained. 

\noindent 
\begin{table}[H]
\begin{minipage}[t]{0.49\columnwidth}%
\noindent \begin{center}
\begin{tabular}{c|c|c|c|}
\multicolumn{1}{c}{} &  & \multicolumn{2}{c|}{Worker}\tabularnewline
\cline{3-4} \cline{4-4} 
\multicolumn{1}{c}{} &  & $\left(20,10\right)$ & $\left(10,20\right)$\tabularnewline
\hline 
\multirow{2}{*}{Firm} & $\left(10,10\right)$ & 0 & 1\tabularnewline
\cline{2-4} \cline{3-4} \cline{4-4} 
 & $\left(20,20\right)$ & 1 & 0\tabularnewline
\hline 
\end{tabular}
\par\end{center}
\noindent \begin{center}
Optimal matching matrix
\par\end{center}%
\end{minipage}%
\begin{minipage}[t]{0.49\columnwidth}%
\noindent \begin{center}
\begin{tabular}{c|c|c|c|}
\multicolumn{1}{c}{} &  & \multicolumn{2}{c|}{Worker}\tabularnewline
\cline{3-4} \cline{4-4} 
\multicolumn{1}{c}{} &  & $10$ & $20$\tabularnewline
\hline 
\multirow{2}{*}{Firm} & $10$ & 0 & .5\tabularnewline
\cline{2-4} \cline{3-4} \cline{4-4} 
 & $20$ & .5 & 0\tabularnewline
\hline 
\end{tabular}
\par\end{center}
\noindent \begin{center}
Optimal matching density
\par\end{center}%
\end{minipage}
\end{table}

\noindent Note that the optimal matching between firms' and workers'
second attributes exhibits negative sorting. Negative sorting is consistent
with negative complementarity and cannot be obtained under strictly
positive complementarity according to Becker's (1973) sorting results.
Based on this observation, econometricians may infer negative complementarity
between firms' and workers' second attributes whereas the underlying
process $\tilde{Q}$ exhibits positive complementarity between these
attributes. 
\end{example}
In order to limit the outcome-relevance of unobserved characteristics,
three key assumptions are adopted in empirical matching literature.
\begin{assumption}
\label{assu:1} There is a large number of agents for each observed
type. 
\end{assumption}
Under the large market assumption, one can focus on the asymptotic
properties of optimal matching. In this context, the small sample
properties of optimal sorting patterns remain an open question. 
\begin{assumption}
\label{assu:2} Matching output function $\tilde{Q}$ can be expressed
as follows:
\noindent \begin{center}
$\tilde{Q}\left(\tilde{x}^{f},\tilde{y}^{w}\right)=Q\left(x^{f},y^{w}\right)+\varepsilon_{f}\left(\tilde{x}^{f},y^{w}\right)+\eta_{w}\left(x^{f},\tilde{y}^{w}\right).$
\par\end{center}

\end{assumption}
Here, the first part of the matching output is determined by observed
attributes. The deterministic matching complementarities are to be
estimated by using empirical matching density. However, notice that
the matching outcome is not determined solely by the deterministic
matching complementarities. The last two parts of the equation represent
idiosyncratic production shocks that affect the matching outcome.
This separability assumption rules out complementarities between unobserved
characteristics. By doing so, it allows for a tractable aggregation
over unobserved characteristics. 
\begin{rem}
\label{rem:1}Consider two firms $f$ and $f^{\prime}$, and two workers
$w$ and $w^{\prime}$ such that $x^{f}=x^{f^{\prime}}=x$ and $y^{w}=y^{w^{\prime}}=y$.
For these four agents, Assumption \ref{assu:2} implies that
\noindent \begin{center}
$\tilde{Q}\left(\tilde{x}^{f},\tilde{y}^{w}\right)+\tilde{Q}\left(\tilde{x}^{f^{\prime}},\tilde{y}^{w^{\prime}}\right)=\tilde{Q}\left(\tilde{x}^{f^{\prime}},\tilde{y}^{w}\right)+\tilde{Q}\left(\tilde{x}^{f},\tilde{y}^{w^{\prime}}\right).$
\par\end{center}

\end{rem}
Remark \ref{rem:1} states that unobserved aggregate output generated
by pairs of couples $\left(f,w\right)$ and $\left(f^{\prime},w^{\prime}\right)$
that have the same observed attributes does not change with a partner-swap
between these two pairs. Assumptions \ref{assu:1} and \ref{assu:2}
make it is possible to link optimal matching between full attributes
to the one between observed attributes when the distribution of $\varepsilon_{f}$
conditional on $x$ and the distribution of $\eta_{w}$ conditional
on $y$ are known. 
\begin{assumption}
\label{assu:3}The idiosynratic production shocks satisfy the following
conditions:

\noindent \textbf{(a)} for all $f\in\mathbb{F}$ such that $x^{f}=x$,
$\boldsymbol{\varepsilon^{f}}=\left\{ \varepsilon_{f}\left(\tilde{x}^{f},y\right)\right\} _{y}$
is drawn from probability distribution $\mathcal{F}_{x}$; and 

\noindent \textbf{(b)} for all $w\in\mathbb{W}$ such that $y^{w}=y$,
$\boldsymbol{\eta^{w}}=\left\{ \eta_{w}\left(x,\tilde{y}^{w}\right)\right\} _{x}$
is drawn from probability distribution $\mathcal{G}_{y}$. 
\end{assumption}
Under these assumptions, the full attribute vector of firm $f$ can
be represented by $\left(x^{f},\boldsymbol{\varepsilon^{f}}\right)$,
where probability distribution of $\varepsilon_{f}\left(\tilde{x}^{f},y\right)$
conditional on $x^{f}=x$ is $\mathcal{F}_{x}$. Similarly, the full
attribute vector of worker $w$ can be represented by $\left(y^{w},\boldsymbol{\eta^{w}}\right)$,
where probability distribution of $\eta_{w}\left(x,\tilde{y}^{w}\right)$
conditional on $y^{w}=y$ is $\mathcal{G}_{y}$. For example, Choo
and Siow (2006) assume that $\mathcal{F}_{x}$ and $\mathcal{G}_{y}$
are standard Gumbel distributions, i.e. $Pr\left\{ \varepsilon_{f}\left(\tilde{x}^{f},y\right)\leq\varepsilon|x^{f}=x\right\} =\exp\left\{ -\exp\left\{ -\varepsilon\right\} \right\} $.
In this section, I consider a simple homoskedastic extension of this
distributional assumption by relaxing Galichon and Salanie's (2010)
assumption. Let the distribution of $\varepsilon_{f}$ conditional
on $x_{w}=x$ be a Gumbel distribution with location parameter $\alpha\left(x\right)$
and scale parameter $\sigma$. Similarly, let the distribution of
$\eta_{w}$ conditional on $y_{w}=y$ be a Gumbel distribution with
location parameter $\beta\left(y\right)$ and scale parameter $\delta$. 
\noindent \begin{center}
$\begin{array}{c}
Pr\left\{ \varepsilon_{f}\left(\tilde{x}^{f},y\right)\leq\varepsilon|x^{f}=x\right\} =\exp\left\{ -\exp\left\{ -\left(\dfrac{\varepsilon-\alpha\left(x\right)}{\sigma}\right)\right\} \right\} \\
\\
Pr\left\{ \eta_{w}\left(x,\tilde{y}^{w}\right)\leq\eta|y^{w}=y\right\} =\exp\left\{ -\exp\left\{ -\left(\dfrac{\eta-\beta\left(y\right)}{\delta}\right)\right\} \right\} 
\end{array}$
\par\end{center}

Galichon and Salanie (2010) establish the uniqueness of optimal matching
density. Furthermore, they present a relationship between optimal
matching density function and double difference of the deterministic
output function, $Q:\mathbb{R}^{K}\times\mathbb{R}^{L}\rightarrow\mathbb{R}$,
when $\alpha\left(x\right)=0$ and $\beta\left(y\right)=0$ for all
$x$ and $y$. Remember that the difference between two independent
random variables following a Gumbel distribution with the same location
and scale parameters follows a logistic distribution with zero location
parameter. Consequently, Galichon and Salanie's (2010) results hold
with type-dependent location parameters as well. 
\begin{lem}
\label{lem:2}The optimal matching density function of observable
attributes $m:\mathbb{R}^{K}\times\mathbb{R}^{L}\rightarrow\left[0,1\right]$
satisfies the following condition for all $x,x^{\prime}\in\mathbb{R}^{K}$
and $y,y^{\prime}\in\mathbb{R}^{L}$\emph{:}

\begin{equation}
\log\left\{ \dfrac{m\left(x,y\right)m\left(x^{\prime},y^{\prime}\right)}{m\left(x^{\prime},y\right)m\left(x,y^{\prime}\right)}\right\} =\left(\sigma+\delta\right)^{-1}\left\{ Q\left(x,y\right)+Q\left(x^{\prime},y^{\prime}\right)-Q\left(x^{\prime},y\right)-Q\left(x,y^{\prime}\right)\right\} .\label{eq:1}
\end{equation}
\end{lem}
Lemma \ref{lem:2} offers a simple relationship between the complementarity
structure and optimal sorting pattern between observable attributes.
Based on this relationship, Siow (2015) derives a semi-parametric
identification strategy\textbf{ }for matching markets in which each
agent has only one observable attribute, i.e. $K=L=1$. He shows that
the deterministic output function is supermodular(submodular) if and
only if the left-hand side of Equation (\ref{eq:1}) is greater(less)
than $1$ for all $x<x^{\prime}$ and $y<y^{\prime}$. I establish
a similar result by using P,N modularity.
\begin{defn}
\label{def:13}Matching density function $m^{\prime}:\mathbb{R}^{K}\times\mathbb{R}^{L}\rightarrow\left[0,1\right]$
is log P,N modular if and only if $\log m^{\prime}$ is P,N modular.
\end{defn}
\begin{prop}
\label{prop:2}Optimal matching density between observable attributes
is log P,N modular if and only if the deterministic output function
is P,N modular.
\end{prop}
\begin{cor}
\label{cor:1}The probability of observing P,N concordant pairs relative
to N,P concordant pairs is higher when the deterministic output function
is P,N modular.
\end{cor}
In the previous section, Proposition \ref{prop:1} states that for
strictly P,N modular deterministic output functions, the fraction
of N,P concordant pairs equals zero when the scale parameters of the
idiosyncratic output components are zero, i.e. every outcome relevant
attribute is observed by econometricians. Proposition \ref{prop:2}
and Corollary \ref{cor:1} assert that for positive values of scale
parameters, positive fraction of N,P concordant pairs may be observed
when the deterministic output function is strictly P,N modular. Furthermore,
the fraction of P,N concordant pairs are higher than the fraction
of N,P concordant pairs when the deterministic output function is
P,N modular. Thus, the matching models in which the agents have unobserved
characteristics offer milder sorting patterns between observed attributes. 

Lemma \ref{lem:2} also allows me to obtain a comparative static result
that links the changes in deterministic complementarities to the optimal
matching density function. Bojilov and Galichon (2015) present comparative
static results regarding the changes in the deterministic complementarities
when (a) the deterministic output function is quadratic, i.e. $Q\left(x,y\right)=\sum_{k}\sum_{l}\theta_{k,l}x_{k}y_{l}$;
and (b) the observable attributes follow Gaussian distributions. The
quadratic functional form assumption offers a simple relationship
between complementarities and model parameters. More specifically,
complementarity between firms' $k^{th}$ and their workers' $l^{th}$
observable attributes is governed by parameter $\theta_{k,l}$ alone.
I offer a similar comparative static result without imposing restrictions
on the matching output and the distributions of the observable attributes.
\begin{defn}
\label{def:14} Deterministic output function \textbf{$Q:\mathbb{R}^{K}\times\mathbb{R}^{L}\rightarrow\mathbb{R}$
}exhibits higher P,N modularity compared to \textbf{$Q^{\prime}:\mathbb{R}^{K}\times\mathbb{R}^{L}\rightarrow\mathbb{R}$
}if and only if, for all P,N concordant $\left(x,y\right)$ and $\left(x^{\prime},y^{\prime}\right)$,
the following condition holds:
\noindent \begin{center}
$Q\left(x,y\right)+Q\left(x^{\prime},y^{\prime}\right)-Q\left(x^{\prime},y\right)-Q\left(x,y^{\prime}\right)\geq Q^{\prime}\left(x,y\right)+Q^{\prime}\left(x^{\prime},y^{\prime}\right)-Q^{\prime}\left(x^{\prime},y\right)-Q^{\prime}\left(x,y^{\prime}\right).$
\par\end{center}

\end{defn}
Here, a P,N modular increase implies (a) an increase in complementarity
between firms' $i^{th}$ and their workers' $j^{th}$ attributes for
all $\left(i,j\right)\in P$, (b) an decrease in complementarity between
firms' $p^{th}$ and their workers' $q^{th}$ attributes for all $\left(p,q\right)\in N$.
In this context, one can use an uneven P,N modular increase to formulate
a skill-biased\footnote{See Lindenlaub (2017) for a parametric examination of skill-biased
changes in the U.S. labor market.} complementarity change, non-parametrically.
\begin{prop}
\label{prop:3}The fraction of P,N concordant pairs relative to the
fraction of N,P concordant pairs under optimal matching rises with
P,N modular increases in the deterministic output function.
\end{prop}
Proposition \ref{prop:3} strengthens the result presented in Corollary
\ref{cor:1}. For any deterministic output function (P,N modular or
otherwise), a change toward P,N modularity increases the fraction
of P,N concordant pairs and decreases the fraction of N,P concordant
pairs. This result allows us to make inference regarding the dynamic
changes in deterministic complementarities without making any functional
form assumptions on the deterministic matching output. In particular,
an increase in the fraction of all P,N concordant pairs relative to
the fraction of N,P concordant pairs is consistent with a P,N modular
increase in the deterministic output function. For parametric purposes,
one can use a quadratic function to parameterize a P,N modular increase
in the deterministic output function.
\begin{cor}
\label{cor:2}\ 

\noindent Function $Q\left(x,y\right)=\sum_{k}\sum_{l}\theta_{k,l}x_{k}y_{l}$
exhibits higher P,N modularity compared to $Q^{\prime}\left(x,y\right)=\sum_{k}\sum_{l}\beta_{k,l}x_{k}y_{l}$
if and only if 

\noindent \textbf{(a) }$\theta_{i,j}\geq\beta_{i,j}$ for all $\left(i,j\right)\in P$;

\noindent \textbf{(b) }$\theta_{i,j}\leq\beta_{i,j}$ for all $\left(i,j\right)\in N$;
and

\noindent \textbf{(c) }$\theta_{i,j}=\beta_{i,j}$ for all $\left(i,j\right)\notin P\cup N$.
\end{cor}
In this framework, the relationship between bivariate complementarities
and bivariate optimal sorting patterns is complicated. That makes
the results presented in Proposition \ref{prop:3} and Corollary \ref{cor:2}
practically useful.
\begin{example}
\label{ex:5}Consider a matching market in which firms and workers
have two binary observable attributes. Assume that $\sigma+\delta=1$
and 
\noindent \begin{center}
$\left[\begin{array}{cc}
\theta_{1,1} & \theta_{1,2}\\
\theta_{2,1} & \theta_{2,2}
\end{array}\right]=\left[\begin{array}{cc}
5 & 0\\
0 & 1
\end{array}\right]$
\par\end{center}

\noindent for quadratic output function $Q\left(x,y\right)=\sum_{k}\sum_{l}\theta_{k,l}x_{k}y_{l}$.
The density functions of firms and workers, $p_{F}$ and $p_{G}$,
are given below:

\noindent 
\begin{table}[H]
\begin{minipage}[t]{0.49\columnwidth}%
\noindent \begin{center}
$p_{F}\left(x\right)=\begin{cases}
.1 & x=\left(0,0\right)\ or\ x=\left(1,1\right)\\
.4 & x=\left(0,1\right)\ or\ x=\left(1,0\right)
\end{cases}$
\par\end{center}%
\end{minipage}%
\begin{minipage}[t]{0.49\columnwidth}%
\noindent \begin{center}
$p_{G}\left(y\right)=\begin{cases}
.4 & y=\left(0,0\right)\ or\ y=\left(1,1\right)\\
.1 & y=\left(0,1\right)\ or\ y=\left(1,0\right)
\end{cases}$
\par\end{center}%
\end{minipage}.
\end{table}

\noindent For this parameterization, I numerically approximate optimal
matching density $m:\left\{ 0,1\right\} ^{2}\times\left\{ 0,1\right\} ^{2}\rightarrow\left(0,1\right)$
as follows by using iterative proportional fitting procedure\footnote{\noindent For details, see Deming and Stephan (1940).}.

\noindent 
\begin{table}[H]
\begin{minipage}[b]{0.49\columnwidth}%
\noindent \begin{center}
\begin{tabular}{c|c|c|c|c|c}
\multicolumn{1}{c}{} &  & \multicolumn{4}{c}{Workers}\tabularnewline
\cline{3-6} \cline{4-6} \cline{5-6} \cline{6-6} 
\multicolumn{1}{c}{} &  & $\left(0,0\right)$ & $\left(0,1\right)$ & $\left(1,0\right)$ & $\left(1,1\right)$\tabularnewline
\hline 
\multirow{4}{*}{Firms} & $\left(0,0\right)$ & .0848 & .0097 & .0013 & .0042\tabularnewline
\cline{2-6} \cline{3-6} \cline{4-6} \cline{5-6} \cline{6-6} 
 & $\left(0,1\right)$ & .2739 & .0848 & .0042 & .0371\tabularnewline
\cline{2-6} \cline{3-6} \cline{4-6} \cline{5-6} \cline{6-6} 
 & $\left(1,0\right)$ & .0371 & .0042 & .0848 & .2739\tabularnewline
\cline{2-6} \cline{3-6} \cline{4-6} \cline{5-6} \cline{6-6} 
 & $\left(1,1\right)$ & .0042 & .0013 & .0097 & .0848\tabularnewline
\end{tabular}
\par\end{center}
\noindent \begin{center}
Optimal matching density function between agents
\par\end{center}%
\end{minipage}%
\begin{minipage}[b]{0.49\columnwidth}%
\noindent \begin{center}
\begin{tabular}{c|c|c|c}
\multicolumn{1}{c}{} &  & \multicolumn{2}{c}{Workers}\tabularnewline
\cline{3-4} \cline{4-4} 
\multicolumn{1}{c}{} &  & 0 & 1\tabularnewline
\hline 
\multirow{2}{*}{Firms} & 0 & .208 & .292\tabularnewline
\cline{2-4} \cline{3-4} \cline{4-4} 
 & 1 & .292 & .208\tabularnewline
\end{tabular}
\par\end{center}
\noindent \begin{center}
Optimal matching density function between the second attributes
\par\end{center}%
\end{minipage}
\end{table}

\noindent Logarithm of the optimal matching density function exhibits
P,N modularity for $P=\left\{ \left(1,1\right),\left(2,2\right)\right\} $
and $N=\left\{ \right\} $, which is consistent with Proposition \ref{prop:2}.
Based on Proposition \ref{prop:2}, one can also predict that the
deterministic output function is P,N modular for $P=\left\{ \left(1,1\right),\left(2,2\right)\right\} $
and $N=\left\{ \right\} $ given the optimal matching density. On
the other hand, the optimal matching density exhibits negative correlation
between the second attributes of firms and workers. Consequently,
an econometrician may predict negative complementarity between the
second observable attributes of firms and workers while the deterministic
output function exhibits positive complementarity between these attributes,
i.e. $\theta_{2,2}=1>0$. 
\end{example}
An analysis based on two univariate marginals of multidimensional
objects is potentially deceptive not only for the binary case. Let
$F:\mathbb{R}^{K}\rightarrow\left[0,1\right]$ and $G:\mathbb{R}^{L}\rightarrow\left[0,1\right]$
be the distribution functions of firms' and workers' observable attributes.
Assume that the deterministic output function is quadratic, i.e. $Q\left(x,y;\boldsymbol{\theta}\right)=\sum_{k}\sum_{l}\theta_{k,l}x_{k}y_{l}$.
Due to Galichon and Salanie (2015), the logarithm of optimal matching
density function is obtained as follows for some $\phi$ and $\varphi$:

\begin{equation}
\log m\left(x,y;\boldsymbol{\theta},F,G\right)=W+\dfrac{Q\left(x,y;\boldsymbol{\theta}\right)+\phi\left(x;\boldsymbol{\theta},F,G\right)+\varphi\left(y;\boldsymbol{\theta},F,G\right)}{\sigma+\delta}\label{eq:2}
\end{equation}

\noindent where 
\begin{onehalfspace}
\noindent \begin{center}
$W=-\log\left\{ \underset{x^{\prime}\in X}{\sum}\underset{y^{\prime}\in Y}{\sum}\exp\left\{ \dfrac{Q\left(x^{\prime},y^{\prime};\boldsymbol{\theta}\right)+\phi\left(x^{\prime};\boldsymbol{\theta},F,G\right)+\varphi\left(y^{\prime};\boldsymbol{\theta},F,G\right)}{\sigma+\delta}\right\} \right\} .$
\par\end{center}
\end{onehalfspace}

\noindent Let $m_{k,l}\left(a,b\right)$ denote the fraction of couples
for which $\left(\clubsuit\right)$ firms' $k^{th}$ attributes equal
$a$ and $\left(\spadesuit\right)$ workers' $l^{th}$ attributes
equal $b$. By using Equation (\ref{eq:2}), the logarithm of bivariate
$k,l$ marginal matching density function can be formulated as follows: 

\begin{equation}
\log m_{k,l}\left(a,b;\boldsymbol{\theta},F,G\right)=\dfrac{\theta_{k,l}}{\sigma+\delta}ab+W+\Lambda_{k,l}\left(a,b\right)\label{eq:3}
\end{equation}

\noindent where
\noindent \begin{center}
$\Lambda_{k,l}\left(a,b\right)=\log\left\{ \underset{x^{\prime}\in\mathbb{R}^{K}}{\sum}\underset{y^{\prime}\in\mathbb{R}^{L}}{\sum}\mathbb{I}\left\{ x_{k}=a,y_{l}=b\right\} \exp\left\{ \dfrac{\phi\left(x^{\prime};\boldsymbol{\theta},F,G\right)+\varphi\left(y^{\prime};\boldsymbol{\theta},F,G\right)+\underset{\left(i,j\right)\neq\left(k,l\right)}{\sum}\theta_{i,j}x_{i}y_{j}}{\sigma+\delta}\right\} \right\} .$
\par\end{center}

\noindent Suppose that a firm's $k^{th}$ attribute may take $f_{k}$
different values: $\left\{ a_{1},\cdots,a_{f_{k}}\right\} $, and
a worker's $l^{th}$ attribute may take $w_{l}$ different values:
$\left\{ b_{1},\cdots b_{w_{l}}\right\} $ such that $a_{i+1}>a_{i}$
for $i=1,\ldots,f_{k}-1$ and $b_{j+1}>b_{j}$ for $j=1\ldots w_{l}-1$.
Let $s_{k,l}$ denote the logarithm of the local odds ratio for the
$i^{th}$ value of a firm's $k^{th}$ attribute and the $j^{th}$
value of a worker's $l^{th}$ attribute: 
\noindent \begin{center}
$s_{k,l}\left(i,j\right)\coloneqq\log m_{k,l}\left(a_{i},b_{j}\right)+\log m_{k,l}\left(a_{i+1},b_{j+1}\right)-\log m_{k,l}\left(a_{i+1},b_{j}\right)-\log m_{k,l}\left(a_{i},b_{j+1}\right).$
\par\end{center}

\noindent Based on Equation (\ref{eq:3}), it is easy to see that 

\begin{equation}
s_{k,l}\left(i,j\right)=\dfrac{\theta_{k,l}}{\sigma+\delta}\left(a_{i+1}-a_{i}\right)\left(b_{j+1}-b_{j}\right)+\Omega\left(i,j\right)\label{eq:4}
\end{equation}

\noindent where $\Omega\left(i,j\right)=\Lambda_{k,l}\left(a_{i},b_{j}\right)+\Lambda_{k,l}\left(a_{i+1},b_{j+1}\right)-\Lambda_{k,l}\left(a_{i},b_{j+1}\right)-\Lambda_{k,l}\left(a_{i+1},b_{j}\right)$. 

Notice that the right-hand side of Equation (\ref{eq:4}) is not only
a function of $\theta_{k,l}$ but also all other model parameters,
i.e. the distributions of observable attributes and other complementarity
parameters. Consequently, an inference based on bivariate sorting
patterns is potentially inconsistent and biased. I address this issue
in a separate paper by proposing a multidimensional dependence class
and devising two cardinal measures of multidimensional dependence.

\section{\label{sec:3}Household Composition and Healthcare Insurance Market}

The relationship between agents' health status and their spouses'
education levels has been well-studied in medical literature. Jaffe
et al. (2005) find the mortality risk among men with cardiovascular
disease is higher for those who are married to less-educated women.
In addition, they also document that the mortality risk among women
with breast cancer is higher for those who are married to less-educated
men. Jaffe et al. (2006) note similar findings for men with cardiovascular
disease, and show that one's wife's education level is a stronger
predictor of her husband\textquoteright s mortality than his own education
level. Kravdal (2008) and Skalická and Kunst (2008) report that one's
mortality risk decreases with former and current spouses' education
levels in Norway. Nilsen et al. (2012) also document a strong association
between spousal education and one's self-rated health in Norway. Brown
et al. (2014) confirm that spousal education is positively associated
with self-rated health in the U.S. 

In literature, educational attainment has also been proven to be an
important predictor of mortality differentials (see Kunst and Mackenbach,
1994; Elo and Preston, 1996; Borrell et al., 1999; Mackenbach et al.,
1999; Manor et al., 2000; Krokstad et al., 2002; Manor et al., 2004).
Since better-educated individuals are more likely to be healthy, an
attraction between better-educated or healthier individuals may generate
a spurious positive association between agents' health status and
their spouses' education levels. This channel may disqualify one's
spouse's education level as a robust predictor of one's health status.
Identifying robust predictors of one's health status is essential
for health insurance carriers. In this section, I examine whether
or not one's spouse's health status and education level are robust
predictors of his/her own health status by using the empirical matching
framework described in the previous section. This empirical exercise
demonstrates how one can apply the aforementioned sorting theory to
address several policy-related questions.

\subsection{\label{subsec:3.1}Data}

In this paper, I use the IPUMS-CPS data series for 2010-2017. The
census data contains individual-level information regarding education
and self-rated health levels for 266,569 couples. 
\noindent \begin{center}
\begin{tabular}{|c|c|c|c|c|c|c|c|c|c}
\multirow{2}{*}{\textbf{2010}} & \multirow{2}{*}{\textbf{2011}} & \multirow{2}{*}{\textbf{2012}} & \multirow{2}{*}{\textbf{2013}} & \multirow{2}{*}{\textbf{2014}} & \multirow{2}{*}{\textbf{2015}} & \multirow{2}{*}{\textbf{2016}} & \multirow{2}{*}{\textbf{2017}} & \multirow{2}{*}{\textbf{TOTAL}} & \tabularnewline
 &  &  &  &  &  &  &  &  & \tabularnewline
\cline{1-9} \cline{2-9} \cline{3-9} \cline{4-9} \cline{5-9} \cline{6-9} \cline{7-9} \cline{8-9} \cline{9-9} 
\multirow{2}{*}{35,643} & \multirow{2}{*}{34,664} & \multirow{2}{*}{33,875} & \multirow{2}{*}{34,023} & \multirow{2}{*}{33,498} & \multirow{2}{*}{33,093} & \multirow{2}{*}{30,754} & \multirow{2}{*}{31,019} & \multirow{2}{*}{266,569} & \tabularnewline
 &  &  &  &  &  &  &  &  & \tabularnewline
\cline{1-9} \cline{2-9} \cline{3-9} \cline{4-9} \cline{5-9} \cline{6-9} \cline{7-9} \cline{8-9} \cline{9-9} 
\multicolumn{10}{c}{Number of Households}\tabularnewline
\end{tabular}
\par\end{center}

\noindent For quantitative purposes, I represent each agent by two
indices (education and health). For the $i^{th}$ couple in the sample,
I denote the woman's attributes by $x_{i}=\left(x_{i,E},x_{i,H}\right)$
, and the man's attributes by $y_{i}=\left(y_{i,E},y_{i,H}\right)$.
\begin{table}[H]
\noindent \begin{centering}
\begin{tabular}{c|c|c|c|c|c|c}
 & \multicolumn{5}{c|}{\textbf{Index}} & \tabularnewline
 & \textbf{1} & \textbf{2} & \textbf{3} & \textbf{4} & \textbf{5} & \tabularnewline
\cline{1-6} \cline{2-6} \cline{3-6} \cline{4-6} \cline{5-6} \cline{6-6} 
\multirow{2}{*}{\textbf{Education}} & \multirow{2}{*}{Less than high school} & \multirow{2}{*}{High school} & \multirow{2}{*}{Some college} & \multirow{2}{*}{College} & \multirow{2}{*}{Post-college} & \tabularnewline
 &  &  &  &  &  & \tabularnewline
\cline{1-6} \cline{2-6} \cline{3-6} \cline{4-6} \cline{5-6} \cline{6-6} 
\multirow{2}{*}{\textbf{Health}} & \multirow{2}{*}{Poor} & \multirow{2}{*}{Fair} & \multirow{2}{*}{Good} & \multirow{2}{*}{Very Good} & \multirow{2}{*}{Excellent} & \tabularnewline
 &  &  &  &  &  & \tabularnewline
\cline{1-6} \cline{2-6} \cline{3-6} \cline{4-6} \cline{5-6} \cline{6-6} 
\end{tabular}
\par\end{centering}
\noindent \centering{}\caption{\label{tab:2}Variables}
\end{table}

\subsection{\label{subsec:3.2}Association Concepts and Measures}

Since the variables of interest are ordinal, I examine the association
between these variables by using a rank-correlation measure (Kruskal's
gamma) for each year. Here, I define key concepts and measures to
calculate association between spouses' education levels and health
status.
\begin{defn}
\label{def:15} The $i^{th}$ and the $j^{th}$ couples exhibit concordance(discordance)
between 

\noindent (a) women's health status and education levels if and only
if $\left(x_{i,H}-x_{j,H}\right)\left(x_{i,E}-x_{j,E}\right)>\left(<\right)0$;

\noindent (b) men's health status and education levels if and only
if $\left(y_{i,H}-y_{j,H}\right)\left(y_{i,E}-y_{j,E}\right)>\left(<\right)0$; 

\noindent (c) men's education levels and their wives' health status
if and only if $\left(x_{i,H}-x_{j,H}\right)\left(y_{i,E}-y_{j,E}\right)>\left(<\right)0$;

\noindent (d) men's health status and their wives' education levels
if and only if $\left(y_{i,H}-y_{j,H}\right)\left(x_{i,E}-x_{j,E}\right)>\left(<\right)0$;
and 

\noindent (e) men's and their wives' health status if and only if
$\left(x_{i,H}-x_{j,H}\right)\left(y_{i,H}-y_{j,H}\right)>\left(<\right)0$.
\end{defn}
To measure the association between agents' health status and education
levels, I calculate the following Kruskal's gamma statistics for each
survey year.
\noindent \begin{center}
\begin{tabular}{cccc}
\multirow{2}{*}{$\Gamma_{H,E}^{W,W}=\dfrac{C_{H,E}^{W,W}-D_{H,E}^{W,W}}{C_{H,E}^{W,W}+D_{H,E}^{W,W}}$} & \multirow{2}{*}{$\Gamma_{H,E}^{M,M}=\dfrac{C_{H,E}^{M,M}-D_{H,E}^{M,M}}{C_{H,E}^{M,M}+D_{H,E}^{M,M}}$} & \multirow{2}{*}{} & \tabularnewline
 &  &  & \tabularnewline
\multirow{2}{*}{$\Gamma_{H,E}^{W,M}=\dfrac{C_{H,E}^{W,M}-D_{H,E}^{W,M}}{C_{H,E}^{W,M}+D_{H,E}^{W,M}}$} & \multirow{2}{*}{$\Gamma_{H,E}^{M,W}=\dfrac{C_{H,E}^{M,W}-D_{H,E}^{M,W}}{C_{H,E}^{M,W}+D_{H,E}^{M,W}}$} & \multirow{2}{*}{$\Gamma_{H,H}^{W,M}=\dfrac{C_{H,H}^{W,M}-D_{H,H}^{W,M}}{C_{H,H}^{W,M}+D_{H,H}^{W,M}}$} & \tabularnewline
 &  &  & \tabularnewline
\end{tabular}
\par\end{center}

\noindent Here, $C_{k,l}^{a,b}$ denotes the fraction of pairs of
couples that exhibits concordance between \emph{a}'s attribute-\emph{k}
and \emph{b}'s attribute-\emph{l} for $a,b\in\left\{ W\left(Women\right),M\left(Men\right)\right\} $
and $k,l\in\left\{ E\left(Education\right),H\left(Health\right)\right\} $.
Similarly, $D_{k,l}^{a,b}$ denotes the fraction of pairs of couples
that exhibits discordance between \emph{a}'s attribute-\emph{k} and
\emph{b}'s attribute-\emph{l} (see Definition \ref{def:15}). Kruskal's
gamma takes values between -1 and 1, inclusively. Values close to
-1 indicate strong negative association, and those close to 1 indicate
strong positive association.

\subsection{Descriptive Statistics}

\subsubsection{Agents' Own Health Status and Education Levels}

As it is illustrated in Table \ref{tab:3}, there is a weak positive
association between agents' own health status and education levels.
The association is slightly more pronounced among women. Tables \ref{tab:10}
and \ref{tab:11} indicate that the lower tail distribution of self-rated
health does not vary substantially by education level. On average,
a level increase in education is associated with .2398 level increase
in self-rated health among women, and .2104 level increase among men. 

\noindent 
\begin{table}[H]
\caption{\label{tab:3}Agents' Own Health Status and Education Levels}

\noindent \centering{}%
\begin{tabular}{|l|l|l|l|l|l|l|l|l|l}
\cline{2-9} \cline{3-9} \cline{4-9} \cline{5-9} \cline{6-9} \cline{7-9} \cline{8-9} \cline{9-9} 
\multicolumn{1}{l|}{} & \multicolumn{8}{c|}{\textbf{Survey Year}} & \tabularnewline
\cline{1-9} \cline{2-9} \cline{3-9} \cline{4-9} \cline{5-9} \cline{6-9} \cline{7-9} \cline{8-9} \cline{9-9} 
\textbf{Statistic} & \textbf{2010} & \textbf{2011} & \textbf{2012} & \textbf{2013} & \textbf{2014} & \textbf{2015} & \textbf{2016} & \textbf{2017} & \tabularnewline
\cline{1-9} \cline{2-9} \cline{3-9} \cline{4-9} \cline{5-9} \cline{6-9} \cline{7-9} \cline{8-9} \cline{9-9} 
\multirow{2}{*}{$\Gamma_{H,E}^{W,W}$} & \multirow{2}{*}{.3121} & \multirow{2}{*}{.2934} & \multirow{2}{*}{.3176} & \multirow{2}{*}{.301} & \multirow{2}{*}{.3044} & \multirow{2}{*}{.2807} & \multirow{2}{*}{.2614} & \multirow{2}{*}{.2707} & \tabularnewline
 &  &  &  &  &  &  &  &  & \tabularnewline
\cline{1-9} \cline{2-9} \cline{3-9} \cline{4-9} \cline{5-9} \cline{6-9} \cline{7-9} \cline{8-9} \cline{9-9} 
\multirow{2}{*}{$\Gamma_{H,E}^{M,M}$} & \multirow{2}{*}{.2708} & \multirow{2}{*}{.263} & \multirow{2}{*}{.2748} & \multirow{2}{*}{.2741} & \multirow{2}{*}{.2563} & \multirow{2}{*}{.2475} & \multirow{2}{*}{.2297} & \multirow{2}{*}{.2486} & \tabularnewline
 &  &  &  &  &  &  &  &  & \tabularnewline
\cline{1-9} \cline{2-9} \cline{3-9} \cline{4-9} \cline{5-9} \cline{6-9} \cline{7-9} \cline{8-9} \cline{9-9} 
\end{tabular}
\end{table}

\subsubsection{Agents' Own Health Status and Their Spouses' Education Levels}

Table \ref{tab:4} shows that there is a weak positive association
between agents' own health status and their spouses' education levels.
The association is slightly more pronounced for men. Tables \ref{tab:12}
and \ref{tab:13} indicate that lower tail distribution of agents'
own health status does not vary substantially by their spouses' education
levels. On average, a level increase in one's spouse's education is
associated with .2189 level increase in one's own health among men,
and .2015 level increase among women. These results suggest that the
association between men's health status and their wives' education
levels is higher than the association between men's health status
and their own education levels. On the other hand, the association
between women's health status and their husbands' education levels
is lower than the association between women's health status and their
own education levels.

\noindent 
\begin{table}[H]
\caption{\label{tab:4}Agents' Own Health Status and Their Spouses' Education
Levels}

\noindent \centering{}%
\begin{tabular}{|l|l|l|l|l|l|l|l|l|l}
\cline{2-9} \cline{3-9} \cline{4-9} \cline{5-9} \cline{6-9} \cline{7-9} \cline{8-9} \cline{9-9} 
\multicolumn{1}{l|}{} & \multicolumn{8}{c|}{\textbf{Survey Year}} & \tabularnewline
\cline{1-9} \cline{2-9} \cline{3-9} \cline{4-9} \cline{5-9} \cline{6-9} \cline{7-9} \cline{8-9} \cline{9-9} 
\textbf{Statistic} & \textbf{2010} & \textbf{2011} & \textbf{2012} & \textbf{2013} & \textbf{2014} & \textbf{2015} & \textbf{2016} & \textbf{2017} & \tabularnewline
\cline{1-9} \cline{2-9} \cline{3-9} \cline{4-9} \cline{5-9} \cline{6-9} \cline{7-9} \cline{8-9} \cline{9-9} 
\multirow{2}{*}{$\Gamma_{H,E}^{W,M}$} & \multirow{2}{*}{.2756} & \multirow{2}{*}{.2565} & \multirow{2}{*}{.2753} & \multirow{2}{*}{.2602} & \multirow{2}{*}{.255} & \multirow{2}{*}{.238} & \multirow{2}{*}{.2331} & \multirow{2}{*}{.244} & \tabularnewline
 &  &  &  &  &  &  &  &  & \tabularnewline
\cline{1-9} \cline{2-9} \cline{3-9} \cline{4-9} \cline{5-9} \cline{6-9} \cline{7-9} \cline{8-9} \cline{9-9} 
\multirow{2}{*}{$\Gamma_{H,E}^{M,W}$} & \multirow{2}{*}{.2771} & \multirow{2}{*}{.2669} & \multirow{2}{*}{.2791} & \multirow{2}{*}{.273} & \multirow{2}{*}{.2731} & \multirow{2}{*}{.2597} & \multirow{2}{*}{.2239} & \multirow{2}{*}{.2522} & \tabularnewline
 &  &  &  &  &  &  &  &  & \tabularnewline
\cline{1-9} \cline{2-9} \cline{3-9} \cline{4-9} \cline{5-9} \cline{6-9} \cline{7-9} \cline{8-9} \cline{9-9} 
\end{tabular}
\end{table}

\subsubsection{Agents' Own and Their Spouses' Health Status}

Table \ref{tab:5} reports a strong positive association between agents'
own and their spouses' health status. Individuals are most likely
to be married to spouses with the same self-rated health (see Table
\ref{tab:14}). This strong positive association has an important
actuarial implication: the risk associated with a two-person family
plan is higher than the aggregate risk associated with two individual
plans. 

\noindent 
\begin{table}[H]
\caption{\label{tab:5}One's Health and One's Spouse's Health}

\noindent \centering{}%
\begin{tabular}{|l|l|l|l|l|l|l|l|l|l}
\cline{2-9} \cline{3-9} \cline{4-9} \cline{5-9} \cline{6-9} \cline{7-9} \cline{8-9} \cline{9-9} 
\multicolumn{1}{l|}{} & \multicolumn{8}{c|}{\textbf{Survey Year}} & \tabularnewline
\cline{1-9} \cline{2-9} \cline{3-9} \cline{4-9} \cline{5-9} \cline{6-9} \cline{7-9} \cline{8-9} \cline{9-9} 
\textbf{Statistic} & \textbf{2010} & \textbf{2011} & \textbf{2012} & \textbf{2013} & \textbf{2014} & \textbf{2015} & \textbf{2016} & \textbf{2017} & \tabularnewline
\cline{1-9} \cline{2-9} \cline{3-9} \cline{4-9} \cline{5-9} \cline{6-9} \cline{7-9} \cline{8-9} \cline{9-9} 
\multirow{2}{*}{$\Gamma_{H,H}^{W,M}$} & \multirow{2}{*}{.7586} & \multirow{2}{*}{.7469} & \multirow{2}{*}{.7328} & \multirow{2}{*}{.7423} & \multirow{2}{*}{.7384} & \multirow{2}{*}{.7486} & \multirow{2}{*}{.7396} & \multirow{2}{*}{.7498} & \tabularnewline
 &  &  &  &  &  &  &  &  & \tabularnewline
\cline{1-9} \cline{2-9} \cline{3-9} \cline{4-9} \cline{5-9} \cline{6-9} \cline{7-9} \cline{8-9} \cline{9-9} 
\end{tabular}
\end{table}

The nation's first and largest private online marketplace for health
insurance, \emph{eHealth, Inc.}, reports average insurance premiums
and deductibles\footnote{https://www.healthcare.gov/glossary/deductible/}
for individual and family plans every year. According to the latest\footnote{https://news.ehealthinsurance.com/\_ir/68/20169/eHealth\%20Health\%20Insurance\%20Price\%20Index\%20Report\%20for\%20the\%202016\%20Open\%20Enrollment\%20Period\%20-\%20October\%202016.pdf}
report, a two-person family insurance costs \$717, whereas an individual
plan costs \$310 to a single man, and \$332\footnote{The gender-premium gap exists despite the fact that gender discrimination
in premium calculations is illegal. 

https://www.legalmatch.com/law-library/article/health-insurance-discrimination-laws.html

https://www.nwlc.org/wp-content/uploads/2015/08/Individual\%20Insurance.pdf} to a single woman. Although per capita insurance premium is higher
for married individuals, insurance policies with lower insurance premiums
also have higher deductibles. According to the same report, a two-person
family plan has \$8,113 deductible. In addition, an individual plan
has \$4,457 deductible for men and \$4,259 deductible for women. Based
on these figures, it is not clear whether or not the insurers assess
the extra risk associated with family plans adequately. The evaluation
of the extra risk implied by the positive association between spouses'
health status requires an in-depth analysis which is beyond the scope
of this paper. 

\subsection{\label{subsec:3.3}Parametric Estimation}

Here, I estimate a parametric matching model to explain the association
patterns presented above. The main goal is to decompose the mechanism
that generates the association patterns into two channels: attraction
and distribution. The empirical framework presented in Section \ref{sec:2}
allows me to dissociate these two effects. Therefore, this decomposition
reveals spurious association patterns. For the purposes of parametric
estimation, I adopt two additional assumptions.
\begin{assumption}
\label{assu:4}Deterministic output function $Q:\left\{ 1,2,3,4,5\right\} ^{2}\times\left\{ 1,2,3,4,5\right\} ^{2}\rightarrow\mathbb{R}$
is quadratic:

\[
Q\left(\boldsymbol{x_{r}},\boldsymbol{y_{c}}\right)=\sum_{k\in\left\{ H,E\right\} }\sum_{l\in\left\{ H,E\right\} }\theta_{k,l}x_{r,k}y_{c,l}
\]

\noindent where $\boldsymbol{x_{r}}$ denotes the attribute vector
of a type-r woman, and $\boldsymbol{y_{c}}$ denotes the attribute
vector of a type-c man for $r,c\in\left\{ 1,...,25\right\} $.
\end{assumption}
\noindent 
\begin{table}[H]
\noindent \begin{centering}
\begin{tabular}{|c|c|c|c|c|c|}
\cline{2-6} \cline{3-6} \cline{4-6} \cline{5-6} \cline{6-6} 
\multicolumn{1}{c|}{} & \multicolumn{5}{c|}{\textbf{Health}}\tabularnewline
\hline 
\textbf{Education} & \textbf{1} & \textbf{2} & \textbf{3} & \textbf{4} & \textbf{5}\tabularnewline
\hline 
\textbf{1} & 1 & 6 & 11 & 16 & 21\tabularnewline
\hline 
\textbf{2} & 2 & 7 & 12 & 17 & 22\tabularnewline
\hline 
\textbf{3} & 3 & 8 & 13 & 18 & 23\tabularnewline
\hline 
\textbf{4} & 4 & 9 & 14 & 19 & 24\tabularnewline
\hline 
\textbf{5} & 5 & 10 & 15 & 20 & 25\tabularnewline
\hline 
\end{tabular}
\par\end{centering}
\caption{Agent's Types}
\end{table}

\begin{assumption}
\label{assu:5}The scale parameters of idiosyncratic output components
add up to 1, i.e. $\sigma+\delta=1$.
\end{assumption}
Under Assumptions \ref{assu:1}-\ref{assu:5}, the probability of
observing a marriage between a type-r woman and a type-c man among
all marriages of type-r women, denoted by $p\left(\boldsymbol{x_{r}},\boldsymbol{y_{c}}\right)$,
can be formulated as follows:

\begin{equation}
p\left(\boldsymbol{x_{r}},\boldsymbol{y_{c}}\right)\coloneqq\dfrac{m\left(\boldsymbol{x_{r}},\boldsymbol{y_{c}}\right)}{\underset{j=1}{\overset{25}{\sum}}m\left(\boldsymbol{x_{r}},\boldsymbol{y_{j}}\right)}=\dfrac{\exp\left\{ Q\left(\boldsymbol{x_{r}},\boldsymbol{y_{c}}\right)+\varphi\left(\boldsymbol{y_{c}}\right)\right\} }{\underset{j=1}{\overset{25}{\sum}}\exp\left\{ Q\left(\boldsymbol{x_{r}},\boldsymbol{y_{j}}\right)+\varphi\left(\boldsymbol{y_{c}}\right)\right\} }.\label{eq:5}
\end{equation}

\noindent An equivalent model can be obtained by choosing the first
type of man as base category:

\begin{doublespace}
\noindent 
\begin{equation}
p\left(\boldsymbol{x_{r}},\boldsymbol{y_{c}}|\boldsymbol{\theta}\right)=\begin{cases}
\dfrac{1}{1+\sum_{j\neq1}\exp\left\{ \underset{k\in\left\{ H,E\right\} }{\sum}\underset{l=\left\{ H,E\right\} }{\sum}\theta_{k,l}x_{r,k}\left(y_{j,l}-y_{j,1}\right)+\left(\varphi\left(\boldsymbol{y_{j}}\right)-\varphi\left(\boldsymbol{y_{1}}\right)\right)\right\} } & c=1\\
\dfrac{\exp\left\{ \underset{k\in\left\{ H,E\right\} }{\sum}\underset{l=\left\{ H,E\right\} }{\sum}\theta_{k,l}x_{r,k}\left(y_{j,l}-y_{j,1}\right)+\left(\varphi\left(\boldsymbol{y_{c}}\right)-\varphi\left(\boldsymbol{y_{1}}\right)\right)\right\} }{1+\sum_{j\neq1}\exp\left\{ \underset{k\in\left\{ H,E\right\} }{\sum}\underset{l=\left\{ H,E\right\} }{\sum}\theta_{k,l}x_{r,k}\left(y_{j,l}-y_{j,1}\right)+\left(\varphi\left(\boldsymbol{y_{j}}\right)-\varphi\left(\boldsymbol{y_{1}}\right)\right)\right\} } & c\neq1
\end{cases}.\label{eq:6}
\end{equation}

\end{doublespace}

Let $f^{m}\left(\boldsymbol{x_{r}}\right)$ denote the fraction of
type-r married women in a sample of $N$ married couples. The likelihood
function is formulated as follows: 

\begin{doublespace}
\noindent 
\[
L_{N}\left(\boldsymbol{\theta}\right)=\prod_{i=1}^{N}m\left(\boldsymbol{x}^{\left(i\right)},\boldsymbol{y}^{\left(i\right)}|\boldsymbol{\theta}\right)=\prod_{i=1}^{N}\prod_{r=1}^{25}\prod_{c=1}^{25}\left\{ m\left(\boldsymbol{x_{r}},\boldsymbol{y_{c}}|\boldsymbol{\theta}\right)\right\} ^{d_{r,c}^{\left(i\right)}}=\prod_{r=1}^{25}\prod_{c=1}^{25}\left\{ f^{m}\left(\boldsymbol{x}_{\boldsymbol{r}}\right)p\left(\boldsymbol{x_{r}},\boldsymbol{y_{c}}|\boldsymbol{\theta}\right)\right\} ^{n_{r,c}}
\]

\end{doublespace}

\noindent where

\noindent $m\left(\boldsymbol{x}^{\left(i\right)},\boldsymbol{y}^{\left(i\right)}|\boldsymbol{\theta}\right)$:
probability of observing the $i^{th}$ couple for parameter vector
$\boldsymbol{\theta}$;

\noindent $m\left(\boldsymbol{x_{r}},\boldsymbol{y_{c}}|\boldsymbol{\theta}\right)$:
probability of observing a marriage between a type-r woman and a type-c
man for parameter vector $\boldsymbol{\theta}$;

\noindent $d_{r,c}^{\left(i\right)}$: binary identifier for the $i^{th}$
couple that equals 1 if the $i^{th}$ marriage is between a type-r
woman and a type-c man; and

\noindent $n_{r,c}$: number of marriages between type-r women and
type-c men. 

\noindent I estimate the complementarities by maximizing the following
log-likelihood function:

\begin{doublespace}
\noindent 
\begin{equation}
\mathcal{L}_{N}\left(\boldsymbol{\theta}\right)=\log L_{N}\left(\boldsymbol{\theta}\right)=\sum_{r=1}^{25}n_{r}\log\left\{ f^{m}\left(\boldsymbol{x}_{\boldsymbol{r}}\right)\right\} +\sum_{r=1}^{25}\sum_{c=1}^{25}n_{r,c}\log\left\{ p\left(\boldsymbol{x_{r}},\boldsymbol{y_{c}}|\boldsymbol{\theta}\right)\right\} \label{eq:7}
\end{equation}

\end{doublespace}

\noindent where $n_{r}$ is the total number of type-r women in the
sample.

The log-likelihood function has 28 parameters: 4 attraction parameters,
$\left\{ \theta_{k,l}\right\} $, and 24 deterministic sympathy parameters,
$\left\{ \varphi\left(\boldsymbol{y_{j}}\right)-\varphi\left(\boldsymbol{y_{1}}\right)\right\} $.
It is a well-known fact that the numerical methods which use a gradient
ascent algorithm are less accurate for high-dimensional parameter
spaces. In order to overcome this problem, I estimate the parameters\footnote{https://en.wikipedia.org/wiki/Simulated\_annealing}\footnote{https://www.mathworks.com/help/gads/simulannealbnd.html}
by using simulated annealing. As it is illustrated in Tables \ref{tab:10}
and \ref{tab:11}, the distributions of the agents vary only slightly
over time. For this reason, I use the entire sample and estimate only
one set of parameters. The estimated parameters are reported in Table
\ref{tab:7}. 

\noindent 
\begin{table}[H]
\caption{\label{tab:7}Estimated Complementarities}

\noindent \centering{}%
\begin{tabular}{c|c|c|c|c}
\multicolumn{1}{c}{} &  & \multicolumn{2}{c|}{Men} & \tabularnewline
\cline{3-4} \cline{4-4} 
\multicolumn{1}{c}{} &  & Health & Education & \tabularnewline
\cline{1-4} \cline{2-4} \cline{3-4} \cline{4-4} 
\multirow{4}{*}{Women} & \multirow{2}{*}{Health} & \multirow{2}{*}{\textbf{\textcolor{blue}{.7625}}} & \multirow{2}{*}{\textcolor{red}{-.0375}} & \tabularnewline
 &  &  &  & \tabularnewline
\cline{2-4} \cline{3-4} \cline{4-4} 
 & \multirow{2}{*}{Education} & \multirow{2}{*}{\textcolor{red}{-.0226}} & \multirow{2}{*}{\textbf{\textcolor{blue}{.5572}}} & \tabularnewline
 &  &  &  & \tabularnewline
\cline{1-4} \cline{2-4} \cline{3-4} \cline{4-4} 
\end{tabular}
\end{table}

These results indicate an attraction between individuals with the
same education levels and health status. The interpretation of these
parameters is a little bit complicated. To interpret $\theta_{H,H}$,
consider four agents $x,x^{\prime},y$, and $y^{\prime}$ such that
(a) $\left(x_{E}-x_{E}^{\prime}\right)=\left(y_{E}-y_{E}^{\prime}\right)=0$,
and (b) $\left(x_{H}-x_{H}^{\prime}\right)\left(y_{H}-y_{H}^{\prime}\right)>0$.
For these agents, one is $e^{\left\{ .7625\left(x_{H}-x_{H}^{\prime}\right)\left(y_{H}-y_{H}^{\prime}\right)\right\} }$
times more likely to observe concordance than discordance between
spouses' health status (see Equation (\ref{eq:1})). In other words,
everything held constant, one is at least $2.1436=e^{.7625}$ times
more likely to observe concordance than discordance between spouses'
health status. Furthermore, the likelihood of concordance rises with
$\left(x_{H}-x_{H}^{\prime}\right)\left(y_{H}-y_{H}^{\prime}\right)$.
The educational attraction parameter can also be interpreted in the
same way: everything held constant, one is at least $1.7458=e^{.5572}$
times more likely to observe concordance relative to discordance between
spouses' education levels.

Contrary to weak positive association between agents' own health status
and their spouses' education levels, the estimation results suggest
a disaffection between healthier individuals and more educated individuals:
everything held constant, one is slightly more likely to observe fewer
educated individuals marrying healthier individuals. Altogether, the
attraction analysis suggests that the weak positive association between
agents' own health status and their spouses' education levels is a
product of three factors: (a) an attraction between better-educated
individuals, (b) an attraction between healthier individuals, and
(c) a weak positive association between agents' health status and
their own education levels. This channel provides a structural support
for strong positive association between one's health and one's spouse's
health, and an empirical justification for the aforementioned premium
gap between family and individual health insurance plans. 

Finally, Tables \ref{tab:8} and \ref{tab:9} show that these inferences
are highly accurate. The structural model described in Section \ref{subsec:3.3}
is a close approximation of the limited household formation process.
As it is illustrated in Table \ref{tab:8}, the predicted association
levels are not very different from the empirical association levels.
Furthermore, Table \ref{tab:9} demonstrates that the statistical
distance between the empirical and the estimated household distributions
is small. Shannon entropy measure indicates that a code to generate
the predicted household distribution has 7.837 average description
length\footnote{https://www.princeton.edu/\textasciitilde cuff/ele201/kulkarni\_text/information.pdf}.
If we use a code to generate the empirical household distribution,
it will have $8.2322=7.837+.3952$ average description length. Therefore,
the efficiency loss associated with using a structural model is only
$100\times\dfrac{.3952}{8.2322}=4.8$ percent. 

\noindent 
\begin{table}[H]
\caption{\label{tab:8}Empirical and Predicted Assocation Levels}

\noindent \centering{}%
\begin{tabular}{|l|lcc|ccc|ccc|ccc|}
\hline 
\multirow{2}{*}{\textbf{Statistics}} & \multirow{2}{*}{} & \multirow{2}{*}{$\Gamma_{H,H}^{W,M}$} & \multirow{2}{*}{} & \multirow{2}{*}{} & \multirow{2}{*}{$\Gamma_{H,E}^{W,M}$} & \multirow{2}{*}{} & \multirow{2}{*}{} & \multirow{2}{*}{$\Gamma_{E,H}^{W,M}$} & \multirow{2}{*}{} &  & \multirow{2}{*}{$\Gamma_{E,E}^{W,M}$} & \tabularnewline
 &  &  &  &  &  &  &  &  &  &  &  & \tabularnewline
\hline 
\multirow{2}{*}{\textbf{Empirical}} & \multirow{2}{*}{} & \multirow{2}{*}{.7439} & \multirow{2}{*}{} & \multirow{2}{*}{} & \multirow{2}{*}{.2546} & \multirow{2}{*}{} & \multirow{2}{*}{} & \multirow{2}{*}{.2638} & \multirow{2}{*}{} &  & \multirow{2}{*}{.6468} & \tabularnewline
 &  &  &  &  &  &  &  &  &  &  &  & \tabularnewline
\hline 
\multirow{2}{*}{\textbf{Predicted}} & \multirow{2}{*}{} & \multirow{2}{*}{.6545} & \multirow{2}{*}{} & \multirow{2}{*}{} & \multirow{2}{*}{.2017} & \multirow{2}{*}{} & \multirow{2}{*}{} & \multirow{2}{*}{.2218} & \multirow{2}{*}{} &  & \multirow{2}{*}{.6041} & \tabularnewline
 &  &  &  &  &  &  &  &  &  &  &  & \tabularnewline
\hline 
\multicolumn{1}{l}{} &  &  & \multicolumn{1}{c}{} &  &  & \multicolumn{1}{c}{} &  &  & \multicolumn{1}{c}{} &  &  & \multicolumn{1}{c}{}\tabularnewline
\end{tabular}
\end{table}

\noindent 
\begin{table}[H]
\caption{\label{tab:9}Difference Between Realized and Estimated Distributions}

\noindent \centering{}%
\begin{tabular}{|l|c|l}
\cline{1-2} \cline{2-2} 
\multirow{2}{*}{\textbf{Kullback-Leibler Divergence $\left(\sum_{i}\hat{m}_{i}\log_{2}\left\{ \hat{m}_{i}/m_{i}\right\} \right)$}} & \multirow{2}{*}{.3952} & \tabularnewline
 &  & \tabularnewline
\cline{1-2} \cline{2-2} 
\multirow{2}{*}{\textbf{Shannon Entropy of Predicted Distribution $\left(-\sum_{i}\hat{m}_{i}\log_{2}\left\{ \hat{m}_{i}\right\} \right)$}} & \multirow{2}{*}{7.837} & \tabularnewline
 &  & \tabularnewline
\cline{1-2} \cline{2-2} 
\multirow{2}{*}{\textbf{Efficiency Loss (\%)}} & \multirow{2}{*}{4.8} & \tabularnewline
 &  & \tabularnewline
\cline{1-2} \cline{2-2} 
\end{tabular}
\end{table}

\section*{Appendix: Proofs}
\begin{proof}
\emph{of Lemma }\ref{lem:1}

\noindent Without loss of generality, assume that $X\coloneqq\textrm{supp}\left(F\right)$
and $Y\coloneqq\textrm{supp}\left(G\right)$ have countably many elements.
For given matching distribution $M$, define $\textrm{vec}M$ as the
density vector defined by $M$. More specifically, the rows of $\textrm{vec}M$
represents the densities of $\left(x,y\right)\in X\times Y$ implied
by $M$. For P,N concordance improving transfer with unit mass $\tau_{P,N}\left(x,y,x^{\prime},y^{\prime};1\right)$,
define P,N transfer vector $t_{P,N}\left(x,y,x^{\prime},y^{\prime}\right)$
such that (a) the row associated with $\left(x,y\right)$ and $\left(x^{\prime},y^{\prime}\right)$
is $1$; (b) the row associated with $\left(x^{\prime},y\right)$
and $\left(x,y^{\prime}\right)$ is $-1$; and (c) the rest of the
rows are $0$. Let $\mathcal{T}\left(P,N\right)$ denote the set of
P,N transfer vectors. 

\noindent It suffices to show that the following statement holds:

\begin{equation}
\textrm{vec}M-\textrm{vec}M^{\prime}=\sum_{t\in\mathcal{T}\left(P,N\right)}\alpha_{t}t\ ,\forall\alpha_{t}\geq0\Leftrightarrow M\succeq_{P,N}M^{\prime}.\label{eq:8}
\end{equation}

\noindent Define $Q\cdot t_{P,N}\left(x,y,x^{\prime},y^{\prime}\right)$
as follows: 
\noindent \begin{center}
$Q\cdot t_{P,N}\left(x,y,x^{\prime},y^{\prime}\right)=Q\left(x,y\right)+Q\left(x^{\prime},y^{\prime}\right)-Q\left(x^{\prime},y\right)-Q\left(x,y^{\prime}\right).$
\par\end{center}

\noindent \ 

\noindent Note that the aggregate output improves with P,N transfers: 

\noindent $Q\cdot t_{P,N}\left(x,y,x^{\prime},y^{\prime}\right)$

\noindent $=Q\left(x_{1},...,x_{K},y_{1},...,y_{L}\right)+Q\left(x_{1}^{\prime},...,x_{K}^{\prime},y_{1}^{\prime},...,y_{L}^{\prime}\right)-Q\left(x_{1},...,x_{K},y_{1}^{\prime},...,y_{L}^{\prime}\right)-Q\left(x_{1}^{\prime},...,x_{K}^{\prime},y_{1},...,y_{L}\right)$

\noindent $=\underset{i=1}{\overset{K}{\sum}}\underset{j=1}{\overset{L}{\sum}}\left\{ \begin{array}{c}
Q\left(x_{1},...,x_{i-1},x_{i},x_{i+1}^{\prime},...,x_{K}^{\prime},y_{1},...,y_{j-1},y_{j},y_{j+1}^{\prime},...,y_{L}^{\prime}\right)\\
+Q\left(x_{1},...,x_{i-1},x_{i}^{\prime},x_{i+1}^{\prime},...,x_{K}^{\prime},y_{1},...,y_{j-1},y_{j}^{\prime},y_{j+1}^{\prime},...,y_{L}^{\prime}\right)\\
-Q\left(x_{1},...,x_{i-1},x_{i},x_{i+1}^{\prime},...,x_{K}^{\prime},y_{1},...,y_{j-1},y_{j}^{\prime},y_{j+1}^{\prime},...,y_{L}^{\prime}\right)\\
-Q\left(x_{1},...,x_{i-1},x_{i}^{\prime},x_{i+1}^{\prime},...,x_{K}^{\prime},y_{1},...,y_{j-1},y_{j},y_{j+1}^{\prime},...,y_{L}^{\prime}\right)
\end{array}\right\} $

\noindent $\geq0$.

\noindent Consequently, it holds that 

\begin{equation}
Q\in\mathbb{C}\left(P,N\right)\Leftrightarrow Q\cdot t\geq0\ \forall t\in\mathcal{T}\left(P,N\right).\label{eq:9}
\end{equation}

\noindent Equation (\ref{eq:8}) holds if and only if $\textrm{vec}M-\textrm{vec}M^{\prime}$
belongs to the convex cone $\mathcal{C}\left(P,N\right)$ generated
by $\mathcal{T}\left(P,N\right)$:
\noindent \begin{center}
$\mathcal{C}\left(P,N\right)=\left\{ \sum_{t\in\mathcal{T}\left(P,N\right)}\alpha_{t}t:\alpha_{t}\geq0,\forall t\in\mathcal{T}\left(P,N\right)\right\} $.
\par\end{center}

\noindent From Equation (\ref{eq:9}), it follows that $\mathbb{C}\left(P,N\right)$
is the dual cone of $\mathcal{C}\left(P,N\right)$. Since $\mathcal{C}\left(P,N\right)$
is convex and closed, $\mathcal{C}\left(P,N\right)$ is the dual cone
of $\mathbb{C}\left(P,N\right)$ due to Luenberger (1969, p:215),
i.e. for $\left\{ \beta_{t}\right\} >0$, it holds that
\noindent \begin{center}
$\sum_{t\in\mathcal{T}\left(P,N\right)}\beta_{t}t\in\mathcal{C}\left(P,N\right)\Leftrightarrow\sum_{t\in\mathcal{T}\left(P,N\right)}\beta_{t}Q\cdot t\geq0\ \forall Q\in\mathbb{C}\left(P,N\right).$
\par\end{center}

\noindent Therefore, $M\succeq_{P,N}M^{\prime}$ if and only if $\textrm{vec}M-\textrm{vec}M^{\prime}\in\mathcal{C}\left(P,N\right)$.
\end{proof}
\begin{claim*}
\ 

\noindent \textbf{(a)} The set of P,N undominated distributions is
a subset of weak P,N assortative matching distributions.

\noindent \textbf{(b)} The set of P,N dominant distributions and global
P,N assortative matching distributions coincide.
\end{claim*}
\begin{proof}
\noindent \textbf{Kakutani Fixed Point Theorem:} Let $A\subseteq\mathbb{R}^{N}$
be a non-empty, compact and convex set; $U:A\mapsto P\left(A\right)$
be a non-empty-valued, convex-valued correspondence with a closed
graph. Correspondence $U:A\mapsto P\left(A\right)$ has a fixed point.

\noindent For a given $M\in\mathcal{M}\left(F,G\right)$ and $P,N$,
define correspondence

\noindent {\small{}
\begin{equation}
U\left(\textrm{vec}M;P,N\right)=\begin{cases}
\textrm{vec}M+\underset{t\in\mathcal{T}\left(P,N\right)/\mathcal{T}\left(N,P\right)}{\sum}\alpha_{t}t+\underset{t\in\mathcal{T}\left(P,N\right)\cap\mathcal{T}\left(N,P\right)}{\sum}\beta_{t}t & ,\ \forall\alpha_{t},\beta_{t}\geq0\ \textrm{and\ for\ some}\ \alpha_{t}>0\\
\textrm{vec}M & oth.
\end{cases}.\label{eq:10}
\end{equation}
}{\small\par}

\noindent It is clear that $U\left(\cdot;P,N\right)$ is non-empty-valued,
convex-valued, and has a closed graph. Therefore, $\exists M^{*}\in\mathcal{M}\left(F,G\right)$
such that $\textrm{vec}M{}^{*}\in U\left(\textrm{vec}M^{*},P,N\right)$
by Kakutani fixed point theorem.

\noindent \textbf{(a)} Let $\textrm{vec}M^{*}$ be a fixed point of
the correspondence described in Equation \ref{eq:10}.Due to Lemma
\ref{lem:1}, every P,N undominated matching distribution corresponds
to a fixed point of the correspondence above. This proves that the
set of P,N undominated distributions is non-empty. By definition,
any distribution violating weak P,N sorting cannot be a fixed point
of the correspondence.

\noindent \textbf{(b)} $\left(\Rightarrow\right)$ Suppose not. Let
$M\in\mathcal{M}\left(F,G\right)$ be P,N dominant and not globally
P,N assortative. Without loss of generality, assume that $\left(1,1\right)\in P$
and $M$ violates $\left(1,1\right)$ positive sorting by assigning
positive mass to $\left(x,y^{\prime}\right),\left(x^{\prime},y\right)$
such that $x_{1}>x_{1}^{\prime}$ and $y_{1}>y_{1}^{\prime}$. For
matching distribution $M^{\prime}\in\mathcal{M}\left(F,G\right)$
satisfying positive sorting between the first attributes, we have
$\int QdM^{\prime}>\int QdM$ when $Q\left(x,y\right)=x_{1}y_{1}$.
Consequently, $M$ is not a P,N dominant distribution.

\noindent $\left(\Leftarrow\right)$ Let $M\in\mathcal{M}\left(F,G\right)$
satisfy global P,N assortativeness. For any $\left(x,y\right),\left(x^{\prime},y^{\prime}\right)\in supp\left(M\right)$,
we have $\left(\clubsuit\right)$ $\left(x_{i}-x_{i}^{\prime}\right)\left(y_{j}-y_{j}^{\prime}\right)\geq0$
for all $\left(i,j\right)\in P$ and $\left(\spadesuit\right)$ $\left(x_{p}-x_{p}^{\prime}\right)\left(y_{q}-y_{q}^{\prime}\right)\geq0$
for all $\left(p,q\right)\in N$. Since every pair of matched couples
under $M$ is P,N weak concordant, any matching distribution $M^{\prime}\in\mathcal{M}\left(F,G\right)$
satisfies the following condition: 

\[
\textrm{vec}M=\textrm{vec}M^{\prime}+\sum_{t\in\mathcal{T}\left(P,N\right)}\alpha_{t}t\textrm{,\ for}\alpha_{t}\geq0.
\]

\noindent Due to Lemma \ref{lem:1}, it immediately follows that any
matching distribution which satisfies global P,N sorting is P,N dominant. 
\end{proof}
\begin{claim*}
If the set of P,N dominant distributions is non-empty, then it coincides
with the set of P,N undominated distributions. 
\end{claim*}
\begin{proof}
\emph{of Claim.}$\left(\Rightarrow\right)$ Trivial. 

\noindent $\left(\Leftarrow\right)$ Suppose not. Let $M\in\mathcal{M}\left(F,G\right)$
be P,N undominated but not P,N dominant. Consider a P,N dominant distribution:
$M^{\prime}\in\mathcal{M}\left(F,G\right)$. Since $M$ is not P,N
dominant and $M^{\prime}$ is, it holds that (a) $\int QdM\leq\int QdM^{\prime}$
for all $Q\in\mathbb{C}\left(P,N\right)$; and (b) there exists output
function $Q\in\mathbb{C}_{+}\left(P,N\right)$ such that $\int QdM<\int QdM^{\prime}$
due to Equation (\ref{eq:9}). Consequently, $M^{\prime}$ strictly
dominates $M$ in P,N modular order. Therefore, $M$ is not a P,N
undominated distribution.
\end{proof}
\begin{proof}
\emph{of Proposition \ref{prop:1}.}

\noindent \textbf{1.a)} If a matching distribution does not satisfy
within-group P,N sorting, then there exist a P,N concordance improving
transfer which is not N,P concordance improving. In other words, the
matching distribution does not satisfy weak P,N sorting.

\noindent \textbf{1.b)}\emph{ }Suppose not. Let $M$ be a solution
to the planner's problem for $Q\in\mathbb{C}_{+}\left(P,N\right)$
which is not a weakly P,N assortative matching distribution. Since
a swap between pair of couples which violates weak P,N sorting strictly
improves the aggregate output, $M$ cannot be a solution.

\noindent \textbf{1.c)}\emph{ }The existence of a weak P,N assortative
distribution immediately follows from the existence of the fixed point
of the correspondence given in Equation \ref{eq:10}.

\noindent Case 1: $P\cup N=\emptyset$.

\noindent Trivial. Every matching distribution $M\in\mathcal{M}\left(F,G\right)$
is associated with the same level of aggregate output. 

\noindent Case 2: $P\cup N\neq\emptyset$.

\noindent Let $M\in\mathcal{M}\left(F,G\right)$ be a solution to
the planner's problem for $Q\in\mathbb{C}\left(P,N\right)$. Suppose
$M$ is not weak P,N assortative distribution. Choose $M^{\prime}\in\mathcal{M}\left(F,G\right)$
such that (i) $\textrm{vec}M^{\prime}\in U\left(\textrm{vec}M;P,N\right)$,
and (ii) $\textrm{vec}M^{\prime}\in U\left(\textrm{vec}M^{\prime};P,N\right)$.
Since $M^{\prime}$ is obtained from $M$ via a sequence of P,N concordance
improving transfers, it cannot worsen the level of aggregate output
for any $Q\in\mathbb{C}\left(P,N\right)$. Consequently, $M^{\prime}$
is at least as good as $M$ for any $Q\in\mathbb{C}\left(P,N\right)$.
Thus, $M^{\prime}$ also solves the planner's problem. By construction,
$M^{\prime}$ is a P,N undominated distribution. Thus $M^{\prime}$
is satisfies weak P,N sorting.

\noindent \textbf{2.a) }Let $M$ be a solution to the planner's problem
for $Q\in\mathbb{C}_{+}\left(P,N\right)$ that does not satisfy global
P,N sorting. Since the set of P,N dominant distributions coincides
with the set of globally P,N assortative matching distributions, and
$M$ is not a globally P,N assortative matching distribution. Thus,
there exists an N,P concordant pair of matched couples under $M$.
In other words, every globally P,N assortative matching distribution
strictly dominates $M$ due to Lemma \ref{lem:1}. Therefore, $M$
cannot solve the planner's problem.

\noindent \textbf{2.b)} The set of P,N dominant distributions coincides
with the set of globally P,N assortative matching distributions. Thus,
every globally P,N assortative matching distribution solves the planner's
problem for all $Q\in\mathbb{C}\left(P,N\right)$. 
\end{proof}
\begin{proof}
\emph{of Proposition \ref{prop:2}. }Immediate from Lemma \ref{lem:2}.
\end{proof}
\begin{proof}
\emph{of Corollary \ref{cor:1}. }Immediate from Proposition \ref{prop:2}.
\end{proof}
\begin{proof}
\emph{of Proposition \ref{prop:3}. }Immediate from Lemma \ref{lem:2}.
\end{proof}
\begin{proof}
\emph{of Corollary \ref{cor:2}. }Immediate from Definition \ref{def:14}.
\end{proof}

\section*{Appendix: Tables}

\noindent 
\begin{table}[H]
{\small{}\caption{\label{tab:10}Distributions of Women's Health Status and Education
Levels by Survey Year}
}{\small\par}
\noindent \centering{}{\small{}}\subfloat{{\small{}}%
\begin{tabular}{|c|c|c|c|c|c|}
\textbf{\small{}2010} & \multicolumn{5}{c|}{\textbf{\small{}Education}}\tabularnewline
\hline 
\textbf{\small{}Health} & \textbf{\small{}1} & \textbf{\small{}2} & \textbf{\small{}3} & \textbf{\small{}4} & \textbf{\small{}5}\tabularnewline
\hline 
\textbf{\small{}1} & {\small{}.0082} & {\small{}.0136} & {\small{}.008} & {\small{}.0026} & {\small{}.0015}\tabularnewline
\hline 
\textbf{\small{}2} & {\small{}.0178} & {\small{}.0374} & {\small{}.0237} & {\small{}.0094} & {\small{}.0042}\tabularnewline
\hline 
\textbf{\small{}3} & \textbf{\small{}.0336} & \textbf{\small{}.1035} & {\small{}.0772} & {\small{}.0423} & {\small{}.0216}\tabularnewline
\hline 
\textbf{\small{}4} & {\small{}.0227} & {\small{}.0968} & \textbf{\small{}.0964} & {\small{}.0754} & {\small{}.0411}\tabularnewline
\hline 
\textbf{\small{}5} & {\small{}.0115} & {\small{}.0541} & {\small{}.0747} & \textbf{\small{}.0788} & \textbf{\small{}.0442}\tabularnewline
\hline 
\multicolumn{1}{c}{} & \multicolumn{1}{c}{} & \multicolumn{1}{c}{} & \multicolumn{1}{c}{} & \multicolumn{1}{c}{} & \multicolumn{1}{c}{}\tabularnewline
\hline 
\end{tabular}}{\small{}\hfill{}}\subfloat{{\small{}}%
\begin{tabular}{|c|c|c|c|c|c|}
\textbf{\small{}2011} & \multicolumn{5}{c|}{\textbf{\small{}Education}}\tabularnewline
\hline 
\textbf{\small{}Health} & \textbf{\small{}1} & \textbf{\small{}2} & \textbf{\small{}3} & \textbf{\small{}4} & \textbf{\small{}5}\tabularnewline
\hline 
\textbf{\small{}1} & {\small{}.0073} & {\small{}.015} & {\small{}.009} & {\small{}.0036} & {\small{}.0015}\tabularnewline
\hline 
\textbf{\small{}2} & {\small{}.0165} & {\small{}.0364} & {\small{}.0241} & {\small{}.0109} & {\small{}.0049}\tabularnewline
\hline 
\textbf{\small{}3} & \textbf{\small{}.0333} & {\small{}.0971} & {\small{}.0745} & {\small{}.0432} & {\small{}.0224}\tabularnewline
\hline 
\textbf{\small{}4} & {\small{}.0214} & \textbf{\small{}.0975} & \textbf{\small{}.1017} & \textbf{\small{}.0781} & {\small{}.0422}\tabularnewline
\hline 
\textbf{\small{}5} & {\small{}.0123} & {\small{}.0538} & {\small{}.0701} & {\small{}.0774} & \textbf{\small{}.0458}\tabularnewline
\hline 
\multicolumn{1}{c}{} & \multicolumn{1}{c}{} & \multicolumn{1}{c}{} & \multicolumn{1}{c}{} & \multicolumn{1}{c}{} & \multicolumn{1}{c}{}\tabularnewline
\hline 
\end{tabular}}{\small\par}
\end{table}

\pagebreak{}

\noindent 
\begin{table}[H]
{\small{}\caption{\label{tab:11}Distributions of Men's Health Status and Education
Levels by Survey Year}
}{\small\par}
\noindent \centering{}{\small{}}\subfloat{{\small{}}%
\begin{tabular}{|c|c|c|c|c|c|}
\textbf{\small{}2010} & \multicolumn{5}{c|}{\textbf{\small{}Education}}\tabularnewline
\hline 
\textbf{\small{}Health} & \textbf{\small{}1} & \textbf{\small{}2} & \textbf{\small{}3} & \textbf{\small{}4} & \textbf{\small{}5}\tabularnewline
\hline 
\textbf{\small{}1} & {\small{}.0127} & {\small{}.0146} & {\small{}.0084} & {\small{}.0046} & {\small{}.0021}\tabularnewline
\hline 
\textbf{\small{}2} & {\small{}.0191} & {\small{}.0376} & {\small{}.0223} & {\small{}.0113} & {\small{}.0064}\tabularnewline
\hline 
\textbf{\small{}3} & \textbf{\small{}.04} & {\small{}.0967} & {\small{}.069} & {\small{}.0441} & {\small{}.0256}\tabularnewline
\hline 
\textbf{\small{}4} & {\small{}.026} & \textbf{\small{}.0981} & \textbf{\small{}.0862} & \textbf{\small{}.0734} & {\small{}.0431}\tabularnewline
\hline 
\textbf{\small{}5} & {\small{}.0158} & {\small{}.061} & {\small{}.0637} & {\small{}.0693} & \textbf{\small{}.049}\tabularnewline
\hline 
\multicolumn{1}{c}{} & \multicolumn{1}{c}{} & \multicolumn{1}{c}{} & \multicolumn{1}{c}{} & \multicolumn{1}{c}{} & \multicolumn{1}{c}{}\tabularnewline
\hline 
\end{tabular}}{\small{}\hfill{}}\subfloat{{\small{}}%
\begin{tabular}{|c|c|c|c|c|c|}
\textbf{\small{}2011} & \multicolumn{5}{c|}{\textbf{\small{}Education}}\tabularnewline
\hline 
\textbf{\small{}Health} & \textbf{\small{}1} & \textbf{\small{}2} & \textbf{\small{}3} & \textbf{\small{}4} & \textbf{\small{}5}\tabularnewline
\hline 
\textbf{\small{}1} & {\small{}.0108} & {\small{}.015} & {\small{}.0098} & {\small{}.0041} & {\small{}.0027}\tabularnewline
\hline 
\textbf{\small{}2} & {\small{}.0185} & {\small{}.0347} & {\small{}.0227} & {\small{}.0125} & {\small{}.0069}\tabularnewline
\hline 
\textbf{\small{}3} & \textbf{\small{}.0386} & {\small{}.0942} & {\small{}.067} & {\small{}.0417} & {\small{}.0237}\tabularnewline
\hline 
\textbf{\small{}4} & {\small{}.0262} & \textbf{\small{}.0997} & \textbf{\small{}.0901} & \textbf{\small{}.0742} & {\small{}.0468}\tabularnewline
\hline 
\textbf{\small{}5} & {\small{}.0153} & {\small{}.058} & {\small{}.0621} & {\small{}.0707} & \textbf{\small{}.0503}\tabularnewline
\hline 
\multicolumn{1}{c}{} & \multicolumn{1}{c}{} & \multicolumn{1}{c}{} & \multicolumn{1}{c}{} & \multicolumn{1}{c}{} & \multicolumn{1}{c}{}\tabularnewline
\hline 
\end{tabular}}{\small\par}
\end{table}

\noindent \pagebreak{}
\begin{table}[H]
{\small{}\caption{\label{tab:12}Distributions of Women's Health Status and Their Husbands'
Education Levels by Survey Year}
}{\small\par}
\noindent \centering{}{\small{}}\subfloat{{\small{}}%
\begin{tabular}{|c|c|c|c|c|c|}
\textbf{\small{}2010} & \multicolumn{5}{c|}{\textbf{\small{}Education}}\tabularnewline
\hline 
\textbf{\small{}Health} & \textbf{\small{}1} & \textbf{\small{}2} & \textbf{\small{}3} & \textbf{\small{}4} & \textbf{\small{}5}\tabularnewline
\hline 
\textbf{\small{}1} & {\small{}.0095} & {\small{}.0125} & {\small{}.0071} & {\small{}.0031} & {\small{}.0017}\tabularnewline
\hline 
\textbf{\small{}2} & {\small{}.0196} & {\small{}.0346} & {\small{}.0211} & {\small{}.0106} & {\small{}.0064}\tabularnewline
\hline 
\textbf{\small{}3} & \textbf{\small{}.0415} & {\small{}.0993} & {\small{}.0692} & {\small{}.0435} & {\small{}.0246}\tabularnewline
\hline 
\textbf{\small{}4} & {\small{}.0286} & \textbf{\small{}.0996} & \textbf{\small{}.0861} & \textbf{\small{}.0737} & {\small{}.0444}\tabularnewline
\hline 
\textbf{\small{}5} & {\small{}.0146} & {\small{}.062} & {\small{}.0661} & {\small{}.0717} & \textbf{\small{}.049}\tabularnewline
\hline 
\multicolumn{1}{c}{} & \multicolumn{1}{c}{} & \multicolumn{1}{c}{} & \multicolumn{1}{c}{} & \multicolumn{1}{c}{} & \multicolumn{1}{c}{}\tabularnewline
\hline 
\end{tabular}}{\small{}\hfill{}}\subfloat{{\small{}}%
\begin{tabular}{|c|c|c|c|c|c|}
\textbf{\small{}2011} & \multicolumn{5}{c|}{\textbf{\small{}Education}}\tabularnewline
\hline 
\textbf{\small{}Health} & \textbf{\small{}1} & \textbf{\small{}2} & \textbf{\small{}3} & \textbf{\small{}4} & \textbf{\small{}5}\tabularnewline
\hline 
\textbf{\small{}1} & {\small{}.0078} & {\small{}.0133} & {\small{}.0085} & {\small{}.0045} & {\small{}.0023}\tabularnewline
\hline 
\textbf{\small{}2} & {\small{}.0188} & {\small{}.0349} & {\small{}.0215} & {\small{}.0109} & {\small{}.0068}\tabularnewline
\hline 
\textbf{\small{}3} & \textbf{\small{}.0399} & {\small{}.0929} & {\small{}.0664} & {\small{}.0429} & {\small{}.0284}\tabularnewline
\hline 
\textbf{\small{}4} & {\small{}.0283} & \textbf{\small{}.1014} & \textbf{\small{}.0909} & \textbf{\small{}.0743} & {\small{}.0461}\tabularnewline
\hline 
\textbf{\small{}5} & {\small{}.0147} & {\small{}.0591} & {\small{}.0644} & {\small{}.0707} & \textbf{\small{}.0505}\tabularnewline
\hline 
\multicolumn{1}{c}{} & \multicolumn{1}{c}{} & \multicolumn{1}{c}{} & \multicolumn{1}{c}{} & \multicolumn{1}{c}{} & \multicolumn{1}{c}{}\tabularnewline
\hline 
\end{tabular}}{\small\par}
\end{table}
\pagebreak{}
\begin{table}[H]
{\small{}\caption{\label{tab:13}Distributions of Men's Health Status and Their Wives'
Education Levels by Survey Year}
}{\small\par}
\noindent \centering{}{\small{}}\subfloat{{\small{}}%
\begin{tabular}{|c|c|c|c|c|c|}
\textbf{\small{}2010} & \multicolumn{5}{c|}{\textbf{\small{}Education}}\tabularnewline
\hline 
\textbf{\small{}Health} & \textbf{\small{}1} & \textbf{\small{}2} & \textbf{\small{}3} & \textbf{\small{}4} & \textbf{\small{}5}\tabularnewline
\hline 
\textbf{\small{}1} & {\small{}.0089} & {\small{}.0172} & {\small{}.0101} & {\small{}.004} & {\small{}.0021}\tabularnewline
\hline 
\textbf{\small{}2} & {\small{}.0174} & {\small{}.0386} & {\small{}.023} & {\small{}.0113} & {\small{}.0063}\tabularnewline
\hline 
\textbf{\small{}3} & \textbf{\small{}.0325} & \textbf{\small{}.0983} & {\small{}.0783} & {\small{}.0438} & {\small{}.0226}\tabularnewline
\hline 
\textbf{\small{}4} & {\small{}.0225} & {\small{}.0954} & \textbf{\small{}.0945} & {\small{}.0743} & {\small{}.0401}\tabularnewline
\hline 
\textbf{\small{}5} & {\small{}.0124} & {\small{}.0558} & {\small{}.074} & \textbf{\small{}.075} & \textbf{\small{}.0415}\tabularnewline
\hline 
\multicolumn{1}{c}{} & \multicolumn{1}{c}{} & \multicolumn{1}{c}{} & \multicolumn{1}{c}{} & \multicolumn{1}{c}{} & \multicolumn{1}{c}{}\tabularnewline
\hline 
\end{tabular}}{\small{}\hfill{}}\subfloat{{\small{}}%
\begin{tabular}{|c|c|c|c|c|c|}
\textbf{\small{}2011} & \multicolumn{5}{c|}{\textbf{\small{}Education}}\tabularnewline
\hline 
\textbf{\small{}Health} & \textbf{\small{}1} & \textbf{\small{}2} & \textbf{\small{}3} & \textbf{\small{}4} & \textbf{\small{}5}\tabularnewline
\hline 
\textbf{\small{}1} & {\small{}.0086} & {\small{}.0172} & {\small{}.0099} & {\small{}.0042} & {\small{}.0025}\tabularnewline
\hline 
\textbf{\small{}2} & {\small{}.0158} & {\small{}.0357} & {\small{}.0251} & {\small{}.0128} & {\small{}.006}\tabularnewline
\hline 
\textbf{\small{}3} & \textbf{\small{}.0316} & {\small{}.0958} & {\small{}.0748} & {\small{}.0446} & {\small{}.0221}\tabularnewline
\hline 
\textbf{\small{}4} & {\small{}.0218} & \textbf{\small{}.0966} & \textbf{\small{}.099} & \textbf{\small{}.0772} & {\small{}.0424}\tabularnewline
\hline 
\textbf{\small{}5} & {\small{}.0132} & {\small{}.0547} & {\small{}.0705} & {\small{}.0744} & \textbf{\small{}.0437}\tabularnewline
\hline 
\multicolumn{1}{c}{} & \multicolumn{1}{c}{} & \multicolumn{1}{c}{} & \multicolumn{1}{c}{} & \multicolumn{1}{c}{} & \multicolumn{1}{c}{}\tabularnewline
\hline 
\end{tabular}}{\small\par}
\end{table}
\pagebreak{}
\begin{table}[H]
{\small{}\caption{\label{tab:14}Distributions of Women's and Their Husbands' Health
Status by Survey Year}
}{\small\par}
\noindent \centering{}{\small{}}\subfloat{{\small{}}%
\begin{tabular}{|c|c|c|c|c|c|}
\textbf{\small{}2010} & \multicolumn{5}{c|}{\textbf{\small{}Men}}\tabularnewline
\hline 
\textbf{\small{}Women} & \textbf{\small{}1} & \textbf{\small{}2} & \textbf{\small{}3} & \textbf{\small{}4} & \textbf{\small{}5}\tabularnewline
\hline 
\textbf{\small{}1} & \textbf{\small{}.0135} & {\small{}.0074} & {\small{}.0069} & {\small{}.0038} & {\small{}.0024}\tabularnewline
\hline 
\textbf{\small{}2} & {\small{}.0096} & \textbf{\small{}.041} & {\small{}.0242} & {\small{}.0111} & {\small{}.0064}\tabularnewline
\hline 
\textbf{\small{}3} & {\small{}.0116} & {\small{}.0287} & \textbf{\small{}.1863} & {\small{}.0338} & {\small{}.0177}\tabularnewline
\hline 
\textbf{\small{}4} & {\small{}.0049} & {\small{}.0125} & {\small{}.0402} & \textbf{\small{}.2493} & {\small{}.0254}\tabularnewline
\hline 
\textbf{\small{}5} & {\small{}.0027} & {\small{}.0071} & {\small{}.0179} & {\small{}.0287} & \textbf{\small{}.2068}\tabularnewline
\hline 
\multicolumn{1}{c}{} & \multicolumn{1}{c}{} & \multicolumn{1}{c}{} & \multicolumn{1}{c}{} & \multicolumn{1}{c}{} & \multicolumn{1}{c}{}\tabularnewline
\hline 
\end{tabular}}{\small{}\hfill{}}\subfloat{{\small{}}%
\begin{tabular}{|c|c|c|c|c|c|}
\textbf{\small{}2011} & \multicolumn{5}{c|}{\textbf{\small{}Men}}\tabularnewline
\hline 
\textbf{\small{}Women} & \textbf{\small{}1} & \textbf{\small{}2} & \textbf{\small{}3} & \textbf{\small{}4} & \textbf{\small{}5}\tabularnewline
\hline 
\textbf{\small{}1} & \textbf{\small{}.0129} & {\small{}.0072} & {\small{}.0083} & {\small{}.0053} & {\small{}.0027}\tabularnewline
\hline 
\textbf{\small{}2} & {\small{}.0105} & \textbf{\small{}.0415} & {\small{}.0217} & {\small{}.0119} & {\small{}.0072}\tabularnewline
\hline 
\textbf{\small{}3} & {\small{}.0111} & {\small{}.0277} & \textbf{\small{}.1794} & {\small{}.0343} & {\small{}.0179}\tabularnewline
\hline 
\textbf{\small{}4} & {\small{}.0051} & {\small{}.012} & {\small{}.0406} & \textbf{\small{}.2584} & {\small{}.0248}\tabularnewline
\hline 
\textbf{\small{}5} & {\small{}.0027} & {\small{}.0069} & {\small{}.0189} & {\small{}.0271} & \textbf{\small{}.2038}\tabularnewline
\hline 
\multicolumn{1}{c}{} & \multicolumn{1}{c}{} & \multicolumn{1}{c}{} & \multicolumn{1}{c}{} & \multicolumn{1}{c}{} & \multicolumn{1}{c}{}\tabularnewline
\hline 
\end{tabular}}{\small\par}
\end{table}

\pagebreak{}

\section*{References\label{sec:References}}

\noindent {[}1{]} Becker, G. S. A theory of marriage: Part i. \emph{Journal
of Political Economy} \textbf{81(4)}, (1973).

\noindent {[}2{]} Belasen, A. R. and Belasen, A. T. Dual effects of
improving doctor-patient communication: patient satisfaction and hospital
ratings. (2018).

\noindent {[}3{]} Belot, M. and Francesconi, M. Dating preferences
and meeting opportunities in mate choice decisions.\emph{ The Journal
of Human Resources }\textbf{42(8)}, 474-508 (2012).

\noindent {[}4{]} Bojilov, R. and Galichon, A. Matching in closed-form:
equilibrium, identification, and comparative statics. \emph{Economic
Theory }\textbf{61(4)}, 587\textendash 609 (2016).

\noindent {[}5{]} Borrell, C., Regidor, E., Arias, L., Navarro, P.,
Puigpinos R., Dominguez, V., and Plansencia, A. Inequalities in mortality
according to educational level in two large Southern European cities.
\emph{International Journal of Epidemiology }\textbf{28(1)}, 58\textendash 63
(1999).

\noindent {[}6{]} Brown, D. C., Hummer, R. A., and Hayward, M. D.
The Importance of Spousal Education for the Self-Rated Health of Married
Adults in the United States \emph{Popul. Res. Policy Rev.} \textbf{33(1)},
127\textendash 151 (2014).

\noindent {[}7{]} Chade, H., Eeckhout J., and Smith, L. Sorting through
search and matching models.\emph{ Journal of Economic Literature}
\textbf{55(2)}, 1-52 (2017).

\noindent {[}8{]} Chiappori, P. A., Oreffice, S., and Quintana-Domeque,
C. Bidimensional matching with heterogeneous preferences: Smoking
in the marriage market. \emph{Journal of European Economic Association}
\textbf{16(1)}, 161\textendash 198 (2017).

\noindent {[}9{]} Chiappori, P. A., Oreffice, S., and Quintana-Domeque,
C. Fatter attraction: anthropometric and socioeconomic matching on
the marriage market.\emph{ Journal of Political Economy} \textbf{120(4)},
659-695 (2012).

\noindent {[}10{]} Chiappori, P. A., McCann, R., and Pass, B. Multidimensional
matching. \emph{Preprint}. (2016).

\noindent {[}11{]} Choo, E., and Siow, A. Who marries whom and why.
\emph{Journal of Political Economy }\textbf{114(1)}, 175-201 (2006).

\noindent {[}12{]} Colangelo, A., Scarsini, M., and Shaked, M. Some
notions of multivariate positive dependence. \emph{Insurance: Mathematics
and Economics} \textbf{37(1)}, 13-26 (2005).

\noindent {[}13{]} Deming, D. J. The growing importance of social
skills in the labor market\emph{. The Quarterly Journal of Economics
}\textbf{132(4)}, 1593-1640 (2017).

\noindent {[}14{]} Domingue, B. W., Fletcher, J., Conley, D., and
Boardman, J. D. Genetic and educational assortative mating among US
adults. \emph{PNAS},\emph{ }\textbf{111(22)}, 7996-8000 (2014).

\noindent {[}15{]} Dupuy, A. and Galichon, A. Personality traits and
the marriage market. \emph{Journal of Political Econom}y \textbf{122(6)},
1271-1319 (2014).

\noindent {[}16{]} Deming, W. E. and Stephan, F. F. On a least squares
adjustment of a sampled frequency table when the expected marginal
totals are known. \emph{Ann. Math. Statist.} \textbf{11(4)}, 427-444
(1940).

\noindent {[}17{]} Elo, I. and Preston S. H. Educational differentials
in mortality: United States, 1979\textendash 85. \emph{Soc. Sci. Med.
}\textbf{42(1)}, 47\textendash 57 (1996).

\noindent {[}18{]} Fletcher, J. M. and Padron, N. A. Heterogeneity
in spousal matching models. (2015).

\noindent {[}19{]} Gabaix X. and Landier, A. Why has CEO pay increased
so much?. \emph{Quarterly Journal of Economics }\textbf{123(1)}, 49\textendash 100
(2008).

\noindent {[}20{]} Gale D. and Shapley, L. S. College admissions and
the stability of marriage. \emph{The American Mathematical Monthly
}\textbf{69(1)}, 9-15 (1962).

\noindent {[}21{]} Galichon, A. and Salanié, B. Matching with trade-offs:
revealed preferences over competing characteristics. (2010).

\noindent {[}22{]} Galichon, A. and Salanié, B. Cupid\textquoteright s
invisible hand: social surplus and identification in matching models\emph{.
Columbia University Academic Commons}. (2015).

\noindent {[}23{]} Gemici, A. and Laufer, S. Marriage and cohabitation.
(2010).

\noindent {[}24{]} Girsberger, E. M., Rinawi, M., and Krapf, M. Wages
and employment: The role of occupational skills. \emph{IZA Discussion
Paper Series }\textbf{11586} (2018).

\noindent {[}25{]} Greenwood, J., Guner, N., Kocharkov, G., and Santos,
C. Marry your like: Assortative mating and income inequality. \emph{American
Economic Review}\textbf{ 104(5)}, 348-53 (2014).

\noindent {[}26{]} Gretsky, N. E., Ostroy, J. M., and Zame W. R. The
nonatomic assignment model. \emph{Econ. Theory}, \textbf{2(1)}, 103\textendash 27.
(1992).

\noindent {[}27{]} Guvenen, F., Kuruscu, B., Tanaka, S., and Wiczer,
D. Multidimensional skill mismatch (2018).

\noindent {[}28{]} Günter, J., Hitsch, G. J., Hortaçsu, A., and Ariely,
D. Matching and sorting in online dating. \emph{American Economic
Review} \textbf{100(1)}, 130-163 (2010).

\noindent {[}29{]} Hugtenburg, J. G., Timmers, L., Elders, P. J.,
Vervloet, M., and Dijk, L. Definitions, variants, and causes of nonadherence
with medication: a challenge for tailored interventions. \emph{Patient
Prefer Adherence} \textbf{7}, 675\textendash 682 (2013).

\noindent {[}30{]} Jaffe, D. H., Eisenbach, Z., Neumark, Y. D., Manor,
O. Does one's own and one's spouse's education affect overall and
cause-specific mortality in the elderly?. \emph{International Journal
of Epidemiology} \textbf{34(6)}, 1409\textendash 1416 (2005).

\noindent {[}31{]} Jaffe, D. H., Eisenbach, Z., Neumark, Y. D., Manor,
O. Effects of husbands\textquoteright{} and wives\textquoteright{}
education on each other\textquoteright s mortality\emph{. Social Science
\& Medicine }\textbf{62(8)}, 2014\textendash 2023 (2006).

\noindent {[}32{]} Jagosh, J., Boudreau, J. D., Steinert, Y., MacDonald,
M. E., and Ingram, L. The importance of physician listening from the
patients\textquoteright{} perspective: Enhancing diagnosis, healing,
and the doctor\textendash patient relationship. \emph{Patient Education
and Counseling} \textbf{85(3)}, 369\textendash 374 (2011).

\noindent {[}33{]} Klofstad, C. A., McDermott, R., and Hatemi, P.
K. The dating preferences of liberals and conservatives. \emph{Political
Behavior} \textbf{35(3)}, 519-538 (2013).

\noindent {[}34{]} Kravdal, Ø. A broader perspective on education
and mortality: Are we influenced by other people's education?. \emph{Social
Science \& Medicine} \textbf{66(3)}, 620-636 (2008).

\noindent {[}35{]} Kremer, M. The o-ring theory of economic development.
\emph{Quarterly Journal of Economics} \textbf{108(3)}, 551\textendash 75
(1993).

\noindent {[}36{]} Krokstad, S., Johnsen, R., and Westin, S. Trends
in health inequalities by educational level in a Norwegian total population
study. \emph{Journal of Epidemiology and Community Health} \textbf{56(5)},
375\textendash 380 (2002).

\noindent {[}37{]} Kunst, A. E. and Mackenbach, J. P. The size of
mortality differences associated with educational level in nine industrialized
countries. \emph{American Journal of Public Health} \textbf{84(6)},
932\textendash 937 (1994).

\noindent {[}38{]} Lindenlaub, I. Sorting multidimensional types:
theory and application. \emph{The Review of Economic Studies} \textbf{84(2)},
718\textendash 789 (2017).

\noindent {[}39{]} Luenberger, D. G. Optimization by vector space
methods\emph{. Wiley} (1969).

\noindent {[}40{]} Muller, A. and Scarsini, M. Fear of loss, inframodularity,
and transfers. \emph{Journal of Economic Theory} \textbf{147(4)},
1490-1500 (2012).

\noindent {[}41{]} Mackenbach, J. P. et al. Socioeconomic inequalities
in mortality among women and among men: an international study. \emph{American
Journal of Public Health} \textbf{89(12)}, 1800\textendash 1806 (1999).

\noindent {[}42{]} Mangasarian, O. L. Uniqueness of solution in linear
programming\emph{. Linear Algebra and its Applications} \textbf{(25)},
151-162 (1979).

\noindent {[}43{]} Manor, O. et al. Educational differentials in mortality
from cardio-vascular disease among men and women: The Israel Longitudinal
Mortality Study\emph{. Annals of Epidemiology} \textbf{(14)}, 453\textendash 460
(2004).

\noindent {[}44{]} Manor, O. et al. Mortality differentials among
women: The Israel Longitudinal Mortality Study. \emph{Social Science
\& Medicine} \textbf{51(8)}, 1175\textendash 1188 (2000).

\noindent {[}45{]} Mazzi, M. A., Rimondini, M., Zee, E., Boerma, W.,
Zimmermann, C., and Bensing, J. Which patient and doctor behaviours
make a medical consultation more effective from a patient point of
view. Results from a European multicentre study in 31 countries. \emph{Patient
Education and Counseling} \textbf{101(10)} (2018).

\noindent {[}46{]} Meyer, M. and Strulovici, B. The supermodular stochastic
ordering (2013).

\noindent {[}47{]} Nelsen, R. B. An introduction to copulas. \emph{Springer
Series in Statistics }(2006)\emph{.}

\noindent {[}48{]} Nilsen, S. M. et al. Education-based health inequalities
in 18,000 Norwegian couples: the Nord-Trøndelag Health Study (HUNT).
\emph{BMC Public Health} \textbf{(12)}, 998 (2012).

\noindent {[}49{]} Scarsini, M. On measures of concordance. \emph{Stochastica}
\textbf{(8)}, 201\textendash 218 (1984).

\noindent {[}50{]} Shapley, L. S. and Shubik, M. The assignment game
I: the core. \emph{International Journal of Game Theory} \textbf{1(1)},
111\textendash 30 (1972).

\noindent {[}51{]} Shaked, M and Shanthikumar, J. G. Supermodular
stochastic orders and positive dependence of random vectors. \emph{Journal
of Multivariate Analysis} \textbf{61(1)}, 86-101 (1997).

\noindent {[}52{]} Shaked, M. and Shanthikumar, J. G. Stochastic order.
\emph{Springer }(2007).

\noindent {[}53{]} Siow, A. Testing Becker's theory of positive assortative
matching. \emph{Journal of Labor Economics} \textbf{33(2)}, 409-44
(2015).

\noindent {[}54{]} Skalicka, V. and Kunst, A. E. Effects of spouses'
socioeconomic characteristics on mortality among men and women in
a Norwegian longitudinal study. \emph{Soc. Sci. Med. }\textbf{66(9)},
2035-47 (2008).

\noindent {[}55{]} Stavropoulou, C. Non-adherence to medication and
doctor-patient relationship: Evidence from a European survey. \emph{Patient
Education and Counseling} \textbf{83(1)}, 7-13 (2011).

\noindent {[}56{]} Szekli, R., Disney, R. L., and Hur, S. MR/GI/1
queues with positively correlated arrival stream. \emph{Journal of
Applied Probability} \textbf{31(2)}, 497-514 (1994).

\noindent {[}57{]} Villani, C. Optimal transport, old and new. \emph{Springer}
(2008).
\end{document}